\documentclass[preprint]{article}

\usepackage[nonatbib]{neurips_2020}

\usepackage[utf8]{inputenc} %
\usepackage[T1]{fontenc}    %
\usepackage{hyperref}       %
\usepackage{xcolor}
\hypersetup{
    colorlinks = true,
 	linkcolor = teal,
 	citecolor = teal,
 	urlcolor = teal
}
\usepackage{url}            %
\usepackage{booktabs}       %
\usepackage{amsfonts}       %
\usepackage{nicefrac}       %
\usepackage{microtype}      %
\usepackage{mdframed}
\usepackage{dsfont}
\usepackage{xspace}
\usepackage{enumitem}
\usepackage[normalem]{ulem}
\usepackage{bm}
\usepackage[ruled,vlined,linesnumbered]{algorithm2e}

\title{%
Neutralizing Self-Selection Bias in\\ Sampling for Sortition
}

\author{%
Bailey Flanigan, \\
Computer Science Department\\
Carnegie Mellon University\\
\And
Paul G\"olz \\
Computer Science Department\\
Carnegie Mellon University\\
\AND
Anupam Gupta\\
Computer Science Department\\
Carnegie Mellon University\\
\And
Ariel D. Procaccia\\
School of Engineering and Applied Sciences\\
Harvard University
}
\usepackage{amssymb,amsmath,mathtools,amsthm}
\usepackage[capitalize]{cleveref}

\usepackage{enumitem} %
\setlist[itemize]{noitemsep, topsep=0pt}

\theoremstyle{plain}
\newtheorem{theorem}{Theorem}
\crefname{theorem}{Theorem}{Theorems}
\newtheorem{lemma}[theorem]{Lemma}
\crefname{lemma}{Lemma}{Lemmas}

\crefname{conjecture}{Conjecture}{Conjectures}

\crefname{corollary}{Corollary}{Corollaries}

\crefname{fact}{Fact}{Facts}

\crefname{proposition}{Proposition}{Propositions}
\theoremstyle{definition}

\crefname{example}{Example}{Examples}

\newcommand{\Pool}{\mathit{Pool}}
\newcommand{\Population}{\mathit{Population}} %
\newcommand{\Recipients}{\mathit{Recipients}}
\newcommand{\Panel}{\mathit{Panel}}

\makeatletter

\newenvironment{customthm}[1]
  {\count@\c@theorem
   \global\c@theorem#1 %
    \global\advance\c@theorem\m@ne
   \theorem}
  {\endtheorem
   \global\c@theorem\count@}

\newenvironment{customlemma}[1]
  {\count@\c@theorem
   \global\c@theorem#1 %
    \global\advance\c@theorem\m@ne
   \lemma}
  {\endlemma
   \global\c@theorem\count@}

\makeatother

\begin{document}
\maketitle

\begin{abstract}
Sortition is a political system in which decisions are made by panels of randomly selected citizens.
The process for selecting a sortition panel is traditionally thought of as uniform sampling without replacement, which has strong fairness properties. 
In practice, however, sampling without replacement is not possible since only a fraction of agents is willing to participate in a panel when invited, and different demographic groups participate at different rates. In order to still produce panels whose composition resembles that of the population, we develop a sampling algorithm that restores close-to-equal representation probabilities for all agents while satisfying meaningful demographic quotas. As part of its input, our algorithm requires probabilities indicating how likely each volunteer in the pool was to participate.
Since these participation probabilities are not directly observable, we show how to learn them, and demonstrate our approach using data on a real sortition panel combined with information on the general population in the form of publicly available survey data. %

\end{abstract}

\section{Introduction} \label{sec:intro} %
What if political decisions were made not by elected politicians but by a randomly selected panel of citizens?
This is the core idea behind \emph{sortition}, a political system originating in the Athenian democracy of the 5th century BC~\cite{VanReybrouck16}.
A \emph{sortition panel} is a randomly selected set of individuals who are appointed to make a decision on behalf of population from which they were drawn. 
Ideally, sortition panels are selected via uniform sampling without
replacement\,---\,that is, if a panel of size $k$ is selected from a
population of size $n$, then each member of the population has a $k/n$
probability of being selected. This system offers appealing
fairness properties for both individuals and subgroups of the
population: First, each individual knows that she has the same
probability of being selected as anyone else, which assures her an
equal say in decision making. The resulting panel is also, in
expectation, \textit{proportionally representative} to all groups in the population: if a group
comprises $x \%$ of the population, they will in expectation comprise
$x \%$ of the panel as well. In fact, if $k$ is large enough,
concentration of measure makes it likely that even a group's \emph{ex post} share of the panel will be close to $x\%$. Both properties stand in contrast to the status quo of electoral democracy, in which the equal influence of individuals and the fair participation of minority groups are often questioned.

Due to the evident fairness properties of selecting decision makers randomly, sortition has seen a recent
surge in popularity around the world. Over the past year, we have
spoken with several nonprofit organizations whose role it is
to sample and facilitate sortition panels~\cite{pc}.
One of these nonprofits, the \emph{Sortition Foundation}, has organized more than 20 panels in about the past year.\footnote{\url{https://www.youtube.com/watch?v=hz2d_8eBEKg} at 8:53.}
Recent high-profile examples of sortition include
the Irish Citizens' Assembly,\footnote{\url{https://2016-2018.citizensassembly.ie/en/}} which led to Ireland's legalization of abortion in 2018,
and the founding of the first permanent sortition chamber of government,\footnote{\url{https://www.politico.eu/article/belgium-democratic-experiment-citizens-assembly/}} which occurred in a regional parliament in the German-speaking community of Belgium in 2019.

The fairness properties of sortition are often presented as we have described them\,---\,in the setting where panels are selected \emph{from the whole population} via uniform sampling without replacement.
As we have learned from practitioners, however, this sampling approach is not applicable in practice due to limited participation:
typically, only between 2 and 5\% of citizens are willing to participate in the panel when contacted.
Moreover, those who do participate exhibit self-selection bias, i.e., they are not representative of the population, but rather skew toward certain groups with certain features.

To address these issues, sortition practitioners introduce additional steps into the sampling process.
Initially, they send a large number of invitation letters to a random subset of the population.
If a recipient is willing to participate in a panel, they can opt into a \emph{pool} of volunteers.
Ultimately, the panel of size $k$ is sampled from the pool.
Naturally, the pool is unlikely to be representative of the population, which means that
uniformly sampling from the pool would yield panels whose demographic composition is unrepresentative of that of the population.
To prevent grossly unrepresentative panels, many practitioners impose quotas on groups based on orthogonal demographic features such as gender, age, or residence inside the country.
These quotas ensure that the ex-post number of panel members belonging to such a group lies within a narrow interval around the proportional share.
Since it is hard to construct panels satisfying a set of quotas, practitioners typically sample using greedy heuristics. One downside of these greedy algorithms is that they may need many attempts to find a valid panel, and thus might take exponential time to produce a valid panel if one exists. More importantly, however, even though these heuristics tend to eventually find valid panels, the probability with which each individual is selected via their selection process is not controlled in a principled way.

By not deliberately controlling individual selection probabilities, existing panel selection procedures fail to achieve basic fairness guarantees to individuals.
Where uniform sampling in the absence of self-selection bias selects each person with equal
probability $k/n$, currently-used greedy algorithms do not
even guarantee a minimum selection probability for members of the
\textit{pool}, let alone fairly distributed probabilities over members of the population. The absence of such a theoretical guarantee has real ramifications in practice: as we show in \Cref{sec:empirical}, in the real-world instance we study, the greedy algorithm selects several population members with probability less than half the magnitude of $k/n$.

This unfairness to individuals, while problematic in its own right, can also lead to unfairness for groups. In particular, current algorithms use quotas to enforce representation of a set of pre-chosen groups delineated by single features, but these quotas do not protect groups defined by \textit{intersections} of these features: for example, proportional representation of women and of young people does not guarantee the proportional representation of young women. By giving some individuals very low probability of selection on the basis of their \textit{combinations} of features, existing algorithms may systematically allocate very low probability to all members of a certain intersectional group, preventing their perspectives from being fairly represented on the panel.

\subsection{Our Techniques and Results}
The main contribution of this paper is a more principled sampling algorithm that, even in the presence of limited participation and self-selection bias, retains the individual fairness properties of uniform sampling in the absence of these challenges, while also allowing the deterministic satisfaction of quotas. In particular, our algorithm satisfies the following desiderata:
\begin{itemize}
    \item[--] \textit{End-to-End Fairness:} The algorithm selects the panel via a process such that all members of the population appear on the panel with probability asymptotically close to $k/n$.
    This also implies that all groups in the population, including those defined by intersections of arbitrarily many features, will be near-proportionally represented in expectation.
    \item[--] \textit{Deterministic Quota Satisfaction:} The selected panel satisfies certain upper and lower quotas enforcing approximate representation for a set of specified features.
    \item[--] \textit{Computational Efficiency:} The algorithm returns a valid panel (or fails) in polynomial time.
\end{itemize}
\textit{End-to-end} fairness refers to the fact that our algorithm is fair to individuals with respect to their probabilities of going from \textit{population} to \emph{panel}, across the intermediate steps of being invited, opting into the pool, and being selected for the panel.
End-to-end fairness can be seen primarily as a guarantee of individual fairness, while proportional representation of all groups in expectation, along with deterministic quota satisfaction, can be seen as two different guarantees of group fairness.  %

The key challenge in satisfying these desiderata is self-selection
bias, which can result in the pool being totally unrepresentative of
the population.
In the worst case, the pool can be so skewed that it
contains no representative panel\,---\,in fact, the pool might not even contain $k$ members.
As a result, no algorithm can produce a valid panel from
every possible pool. However, we are able to give an algorithm that succeeds with high probability, under weak
assumptions mainly relating the number of invitation
letters sent out to $k$ and the minimum participation probability over all agents.

Crucially, any sampling algorithm that gives (near-)equal selection
probability to all members of the population must reverse the
self-selection bias occurring in the formation of the pool. We formalize this self-selection bias by assuming
that each agent $i$ in the population agrees to join the pool with some positive participation probability $q_i$ when invited.
If these $q_i$ values are known for all members of
the pool, our sampling algorithm can use them to neutralize
self-selection bias.
To do so, our algorithm selects agent $i$ for the panel with a probability (close to) proportional to $1/q_i$, conditioned on $i$ being
in the pool.
This compensates for agents' differing likelihoods of entering the pool, thereby giving all agents an equal end-to-end probability.
On a given pool, the algorithm assigns marginal selection probabilities to every agent in the pool.
Then, to find a distribution over valid panels that implements these marginals, the algorithm randomly rounds a linear program using techniques based on discrepancy theory.
Since our approach aims for a fair \emph{distribution} of valid panels rather than just a single panel, we can give probabilistic fairness guarantees.

As we mentioned, our theoretical and algorithmic results, presented in \cref{sec:theory}, take the probabilities $q_i$ of all pool members $i$ as given in the input.
While these values are not observed in practice, we show in \cref{sec:learning} that they can be estimated from available data.
We cannot directly train a classifier predicting participation, however, because practitioners collect data only on those who \textit{do} join the pool, yielding only positively labeled data.
In place of a negatively labeled control group, we use publicly available survey data, which is unlabeled (i.e., includes no information on whether its members would have joined the pool).
To learn in this more challenging setting, we use techniques from \emph{contaminated controls}, which combine the pool data with the unlabeled sample of the population to learn a predictive model for agents' participation probabilities.
In \cref{sec:empirical}, we use data from a real-world sortition panel to show that plausible participation probabilities can be learned and that the algorithm produces panels that are close to proportional across features.
For a synthetic population produced by extrapolating the real data, we show that our algorithm obtains fair end-to-end probabilities.

\subsection{Related Work}
Our work is broadly related to existing literature on fairness in the areas of \emph{machine learning}, \emph{statistics}, and \emph{social choice}.
Through the lens of fair machine learning, our quotas can be seen as enforcing approximate statistical fairness for protected groups, and our near-equal selection probability as a guarantee on individual fairness.
Achieving simultaneous group- and individual-level fairness is a commonly discussed goal in fair machine learning~\cite{binns2020apparent,golz2019paradoxes, hu2020fair}, but one that has proven somewhat elusive.
To satisfy fairness constraints on orthogonal protected groups, we draw upon techniques from discrepancy theory~\cite{bansal2019,BF81}, which we hope to be more widely applicable in this area.

Our paper addresses self-selection bias, %
which is routinely faced in statistics and usually addressed by sample reweighting.
Indeed, our sampling algorithm can be seen as a way of reweighting the pool members under the constraint that weights must correspond to the marginal probabilities of a random distribution.
While reweighting is typically done by the simpler methods of post-stratification, calibration~\cite{HE91}, and sometimes regression~\cite{PEAS}, we use the more powerful tool of learning with contaminated controls~\cite{lancaster1996case,ward2009presence} to determine weights on a more fine-grained level.

Our paper can also be seen as a part of a broader movement towards statistical approaches in social choice~\cite{MX18,MPS+19,SPX14}.
The problem of selecting a representative sortition panel can be seen as a fair division problem, in which $k$ indivisible copies of a scarce resource must be randomly allocated such that an approximate version of the proportionality axiom is imposed.
Our group fairness guarantees closely resemble the goal of apportionment, in which seats on a legislature are allocated to districts or parties such that each district is proportionally represented within upper and lower quotas~\cite{BY10,BGP+20,Grimmett04}.

So far, only few papers in computer science and statistics directly address sortition~\cite{BGP19,ST13, WX12}. Only one of them~\cite{BGP19} considers, like us, how to sample a representative sortition panel. Unfortunately, their stratified sampling algorithm assumes that all agents are willing to participate, which, as we address in this paper, does not hold in practice.

\section{Model} \label{sec:model}

\noindent\textbf{Agents.}
Let $N$ be a set of $n$ agents, constituting the underlying population. Let $F$ be a set of \textit{features}, where feature $f \in F$ is a function $f : N \to V_f$, mapping the agents to a set $V_f$ of possible values of feature $f$. For example, for the feature $\textit{gender}$, we could have $V_{\textit{gender}} = \{\text{\textit{male, female, non-binary}}\}$. %
Let the \textit{feature-value pairs} be $\bigcup_{f \in F} \{(f,v)
\mid v \in V_f\}$. In our example, the feature-value pairs are $(\textit{gender},\textit{male})$, $(\textit{gender},\textit{female}),$ and $(\textit{gender},\textit{non-binary})$.
Denote the number of agents with a particular feature-value pair $(f, v)$ by $n_{f, v}$.

Each agent $i \in N$ is described by her \textit{feature vector} $F(i) := \{ (f, f(i)) \mid f \in F\}$, the set of all feature-value pairs pertaining to this agent. Building on the example instance, suppose we add the feature \textit{education-level}, so $F = \{\textit{gender}, \textit{education level}\}$. If \textit{education level} can take on the values \textit{college} and \textit{no college}, a college-educated woman would have the feature-vector $\{(\textit{gender}, \textit{female}), (\textit{education level}, \textit{college})\}$.

\noindent\textbf{Panel Selection Process.}
Before starting the selection process, organizers of a sortition panel must commit to the panel's parameters.
First, they must choose the number of \emph{recipients} $r$ who will
be invited to potentially join the panel, and the required \emph{panel
  size} $k$. Moreover, they must choose a set of features $F$ and values $\{V_f\}_{f \in F}$ over which quotas will be imposed.
Finally, for all feature-value pairs $(f,v)$, they must choose a \emph{lower
quota} $\ell_{f,v}$ and an \emph{upper quota} $u_{f,v}$, implying that the
eventual panel of $k$ agents must contain \textit{at least}
$\ell_{f,v}$ and \textit{at most} $u_{f,v}$ agents with value $v$ for feature $f$. Once these parameters are fixed, the panel selection process proceeds in three steps:
\begin{center}
     \large
    \emph{population} \Large $\ \xrightarrow[]{\ \text{\textbf{STEP 1}}\ }$  \large \ \emph{recipients} \  \Large $\ \xrightarrow[]{\ \text{\textbf{STEP 2}}\ }$  \large \ \emph{pool} \  \Large $\xrightarrow[]{\ \text{\textbf{STEP 3}}\ }$ \large \ \emph{panel}
\end{center}

In \textbf{STEP 1}, the organizer of the panel sends out $r$ letters, inviting a subset of the population\,---\,sampled with equal probability and without replacement\,---\,to volunteer for serving on the panel. We refer to the random set of agents who receive these letters as $\Recipients$. Only the agents in $\Recipients$ will have the opportunity to advance in the process toward being on the panel.

In \textbf{STEP 2}, each letter recipient may respond affirmatively to the invitation, thereby opting into the pool of agents from which the panel will be chosen.
These agents form the random set $\Pool$, defined as the set of agents who received a letter and agreed to serve on the panel if ultimately chosen.
We assume that each agent $i$ joins the pool with some \emph{participation probability} $q_i > 0$. Let $q^*$ be the lowest value of $q_i$ across all agents $i \in N$.
A key parameter of an instance is $\alpha \coloneqq q^* \, r / k$, which measures how large the number of recipients is relative to the other parameters.
Larger values of $\alpha$ will allow us the flexibility to satisfy stricter quotas.

In \textbf{STEP 3}, the panel organizer runs a \emph{sampling
  algorithm}, which selects the panel from the pool. This panel, denoted as the set $\Panel$, must be of size $k$ and satisfy the predetermined quotas for all feature-value pairs. The sampling algorithm may also fail without producing a panel.
  
  We consider the first two steps of the process to be fully prescribed. The focus of this paper is to develop a sampling algorithm for the third step that satisfies the three desiderata listed in the introduction: end-to-end fairness, deterministic quota satisfaction, and computational efficiency.

\section{Sampling Algorithm} \label{sec:theory} %
In this section, we give an algorithm which ensures, under natural assumptions, that every agent ends up on the panel with probability at least $\big(1 - o(1)\big) \, k / n$ as $n$ goes to infinity.\footnote{We allow $k \geq 1$ and $r \geq 1$ to vary arbitrarily in $n$ and assume that the feature-value pairs are fixed.}
Furthermore, the panels produced by this algorithm satisfy non-trivial quotas, which ensure that the ex-post representation of each feature-value pair cannot be too far from being proportional.

Our algorithm proceeds in two phases: \textit{I. assignment of marginals}, during which the algorithm assigns a marginal selection probability to every agent in the pool, and \textit{II. rounding of marginals}, in which the marginals are dependently rounded to $0/1$ values, the agents' indicators of being chosen for the panel.
As we discussed previously, our algorithm succeeds only with high probability, rather than deterministically; it may fail in phase~I if the desired marginals do not satisfy certain conditions.
We refer to pools on which our algorithm succeeds as \textit{good pools}.
A good pool, to be defined precisely later, is one that is highly representative of the population\,---\,that is, its size and the prevalence of all feature values within it are close to their respective expected values.
We leave the behavior of our algorithm on bad pools unspecified: while the algorithm may try its utmost on these pools, we give no guarantees in these cases, so the probability of representation guaranteed to each agent must come only from good pools and valid panels.
Fortunately, under reasonable conditions, we show that the pool will be good with high probability.
When the pool is good, our algorithm always succeeds, meaning that our algorithm is successful overall with high probability.

Our algorithm satisfies the following theorem, guaranteeing close-to-equal end-to-end selection probabilities for all members of the population as well as the satisfaction of quotas. 
\begin{theorem} \label{thm:main}
Suppose that $\alpha \to \infty$ and $n_{f,v} \geq n/k$ for all feature-value pairs $f,v$.
Consider a sampling algorithm that, on a good pool, selects a random panel, $\Panel$, via the randomized version of \cref{lem:beckfiala_random}, and else does not return a panel. 
This process satisfies, for all $i$ in the population, that
    \[ \mathbb{P}[i \in \Panel] \geq (1 - o(1)) \, k/n. \]
All panels produced by this process  satisfy the quotas $\ell_{f,v} \coloneqq (1 - \alpha^{-.49}) \, k \, n_{f,v} / n - |F|$ and $u_{f,v} \coloneqq (1 + \alpha^{-.49}) \, k \, n_{f,v} / n + |F|$ for all feature-value pairs $f,v$.
\end{theorem}

The guarantees of the theorem grow stronger as the parameter $\alpha =
q^* \, r/k$ tends toward infinity, i.e., as the number $r$ of
invitations grows. Note that, since $r \leq n$, this assumption requires that $q^* \gg k/n$.
We defer all proofs to \cref{sec:app_theory} and discuss the preconditions in
\cref{sec:preconditions}.

\subsection{Algorithm Part I: Assignment of Marginals} \label{sec:theory-comp_marginals}
To afford equal probability of panel membership to each agent $i$, we would like to select agent $i$ with probability inversely proportional to her probability $q_i$ of being in the pool.
For ease of notation, let $a_i \coloneqq 1/q_i$ for all $i$.
Specifically, for agent $i$, we want $\mathbb{P}[i \in \Panel \mid i \in \Pool]$ to be proportional to $a_i$.
Achieving this exactly is tricky, however, because each agent's \textit{selection probability} from pool $P$, call it $\pi_{i,P}$, must depend on those of all other agents in the pool, since their marginals must add to the panel size $k$.
Thus, instead of reasoning about an agent's probability across all possible pools at once, we take the simpler route of setting agents' selection probabilities for each pool separately, guaranteeing that $\mathbb{P}[i \in \Panel \mid i \in P]$ is proportional to $a_i$ across all members $i$ of a good pool $P$.
For any good pool $P$, we select each agent $i \in P$ for the panel with probability 
\[\textstyle{\pi_{i, P} \coloneqq k \, a_i / \sum_{j \in P} a_j.} \]
Note that this choice ensures that the marginals always sum up to $k$.

\noindent\textbf{Definition of Good Pools.}
For this choice of marginals to be reasonable and useful for giving end-to-end guarantees, the pool $P$ must satisfy three conditions, whose satisfaction defines a \emph{good pool} $P$.
First, the marginals do not make much sense unless all $\pi_{i,P}$ lie in $[0,1]$:
\begin{equation}
    0 \leq \pi_{i, P} \leq 1 \quad \forall i \in P. \label{eq:good1}
\end{equation}
Second, the marginals summed up over all pool members of a feature-value pair $f,v$ should not deviate too far from the proportional share of the pair:
\begin{equation}
    \textstyle{(1 - \alpha^{-.49})\, k \, n_{f, v} / n \leq \sum_{i \in P : f(i) = v} \pi_{i, P} \leq (1 + \alpha^{-.49}) \, k \, n_{f, v} / n \quad \forall f, v.} \label{eq:good2}
\end{equation}
Third, we also require that the term $\sum_{i \in P} a_i$ is not much larger than $\mathbb{E}[\sum_{i \in \Pool} a_i] = r$, which ensures that the $\pi_{i,P}$ do not become to small:
\begin{equation}
    \textstyle{\sum_{i \in P}a_i \leq r/(1-\alpha^{-.49})}. \label{eq:good3}
\end{equation}
Under the assumptions of our theorem, pools are good with high probability, even if we condition on any agent $i$ being in the pool:
\begin{lemma} \label{lem:goodpools}
Suppose that $\alpha \to \infty$ and $n_{f,v} \geq n/k$ for all $f,v$.
Then, for all agents $i \in \Population$, $\mathbb{P}[\text{$\Pool$ \emph{is good}} \mid i \in \Pool] \to 1$.
\end{lemma}
Note that only constraint~\eqref{eq:good1} prevents Phase~II of the algorithm from running; the other two constraints just make the resulting distribution less useful for our proofs.
In practice, if it is possible to rescale the $\pi_{i,P}$ and cap them at $1$ such that their sum is $k$, running phase~II on these marginals seems reasonable.

\subsection{Algorithm Part II: Rounding of Marginals}
The proof of \cref{thm:main} now hinges on our ability to implement the chosen $\pi_{i,P}$ for a good pool $P$ as marginals of a distribution over panels.
This phase can be expressed in the language of randomized dependent rounding: we need to define random variables $X_i = \mathds{1}\{i \in \Panel\}$ for each $i \in \Pool$ such that $\mathbb{E}[X_i] = \pi_{i,P}$.
This difficulty of this task stems from the ex-post requirements on the pool, which require that $\sum_i X_i = k$ and that $\sum_{i:f(i)=v} X_i$ is close to $k \, n_{f,v} / n$ for all feature-value pairs $f,v$.
While off-the-shelf dependent rounding~\cite{CVZ10} can guarantee the marginals and the sum-to-$k$ constraint, it cannot simultaneously ensure small deviations in terms of the representation of all $f,v$.

Our algorithm uses an iterative rounding procedure based on a celebrated theorem by Beck and Fiala~\cite{BF81}.
We sketch here how to obtain a deterministic rounding satisfying the ex-post constraints; the argument can be randomized using results by Bansal~\cite{bansal2019} or via column generation~(\cref{sec:beckfialarandom}).\footnote{Bansal~\cite{bansal2019} gives a black-box polynomial-time method
for randomizing our rounding procedure. We found
column-generation-based algorithms to be faster in practice, with guarantees
that are at least as tight.}
The iterated rounding procedure manages a variable $x_i \in [0, 1]$ for each $i \in \Pool$, which is initialized as $\pi_{i,P}$.
As the $x_i$ are repeatedly updated, more of them are fixed as either 0 or 1 until the $x_i$ ultimately correspond to indicator variables of a panel.
Throughout the rounding procedure, it is preserved that $\sum_i x_i = \sum_i \pi_{i,P} = k$, and the equalities $\sum_{i : f(i)=v} x_i = \sum_{i : f(i)=v} \pi_{i,P}$ are preserved until at most $|F|$ variables $x_i$ in the sum are yet to be fixed.
As a result, the final panel has exactly $k$ members, and the number of members from a feature-value pair $f,v$ is at least $\sum_{i : f(i)=v} \pi_{i,P} - |F| \geq (1 - \alpha^{-.49}) \, k \, n_{f,v}/n - |F|$ (symmetrically for the upper bound).\footnote{Observe that our Beck-Fiala-based rounding procedure only increases the looseness of the quotas by a constant additive term beyond the losses to concentration.
The concentration properties of standard dependent randomized rounding
do not guarantee such a small gap with high probability.
Moreover, our bound does not directly depend on the number of quotas (i.e., twice the number of feature-value pairs) but only depends on the number of features, which are often much fewer.}
As we show in \cref{sec:app_bf},
\begin{lemma}
    \label{lem:beckfiala_random}
    There is a polynomial-time sampling algorithm that, given a good pool $P$, produces a random panel $\Panel$ such that (1) $\mathbb{P}[i \in \Panel] = \pi_{i,P}$ for all $i \in P$, (2) $|\Panel| = k$, and (3) $\sum_{i : f(i)=v} \pi_{i,P} - |F| \leq |\{i \in \Panel \mid f(i)=v\}| \leq \sum_{i : f(i)=v} \pi_{i,P} + |F|$.
\end{lemma}
Our main theorem follows from a simple argument combining \cref{lem:goodpools,lem:beckfiala_random} (\cref{sec:app_thmmain}).

While the statement of \cref{thm:main} is asymptotic in the growth of $\alpha$, the same proof gives bounds on the end-to-end probabilities for finite values of $\alpha$. If one wants bounds for a specific instance, however, bounds uniquely in terms of $\alpha$ tend to be loose, and one might want to relax Condition~\eqref{eq:good2} of a good pool in exchange for more equal end-to-end probabilities. In this case, plugging the specific values of $n,r,k,q^*,n_{f,v}$ into the proof allows to make better trade-offs and to extract sharper bounds.

\section{Learning Participation Probabilities}
\label{sec:learning}
The algorithm presented in the previous section relies on knowing $q_i$ for all agents $i$ in the pool.
While these $q_i$ are not directly observed, we can estimate them from data available to practitioners.

First, we assume that an agent $i$'s participation probability $q_i$ is a function of her feature vector $F(i)$.
Furthermore, we assume that $i$ makes her decision to participate through a specific generative model known as \emph{simple independent action}~\cite[as cited in \cite{Weinberg86}]{Finney71}. %
First, she flips a coin with probability $\beta_0$ of landing on heads.
Then, she flips a coin for each feature $f \in F$, where her coin pertaining to $f$ lands on heads with probability $\beta_{f,f(i)}$.
She participates in the pool if and only if all coins she flips land on heads, leading to the following functional dependency:
\[\textstyle{q_i = \beta_0 \, \prod_{f \in F} \beta_{f,f(i)}}.\]
We think of $1 - \beta_{f,v}$ as the probability that a reason specific to the feature-value pair $f,v$ prevents the agent from participating, and of $1 - \beta_0$ as the baseline probability of her not participating for reasons independent of her features.
The simple independent action model assumes that these reasons occur independently between features, and that the agent participates iff none of the reasons occur.

If we had a representative sample of agents\,---\,say, the recipients of the invitation letters\,---\,labeled according to whether they decided participate (``positive'') or not (``negative''), learning the parameters $\beta$ would be straightforward.
However, sortition practitioners only have access to the features of those who enter the pool, and not of those who never respond.
Without a control group, it is impossible to distinguish a feature that is prevalent in the population and associated with low participation rate from a rare feature associated with a high participation rate.
Thankfully, we can use additional information: in place of a negatively-labeled control group, we use a \emph{background sample}\,---\,a dataset containing the features for a uniform sample of agents, but without labels indicating whether they would participate.
Since this control group contains both positives and negatives, this setting is known as \textit{contaminated controls}.
A final piece of information we use for learning is the fraction $\overline{q} \coloneqq |\Pool|/r$, which estimates the mean participation probability across the population.
In other applications with contaminated controls, including $\overline{q}$ in the estimation increased model identifiability~\cite{ward2009presence}.

To learn our model, we apply methods for maximum likelihood estimation~(MLE) with contaminated controls introduced by Lancaster and Imbens~\cite{lancaster1996case}.
By reformulating the simple independent action model in terms of the logarithms of the $\beta$ parameters, their estimation (with a fixed value of $\overline{q}$) reduces to maximizing a concave function.
\begin{theorem}\label{thm:concave} The log-likelihood function for the simple independent action model under contaminated controls is concave in the model parameters.
\end{theorem}
By this theorem, proven in \cref{sec:app_learning}, we can directly and efficiently estimate $\beta$.
Logistic models, by contrast, require more involved techniques for efficient estimation~\cite{ward2009presence}.

\section{Experiments} \label{sec:empirical}
\label{sec:empir_data}
\noindent\textbf{Data.} We validate our $q_i$ estimation and sampling algorithm on pool data from \emph{Climate Assembly UK},\footnote{\url{https://www.climateassembly.uk/}} a national-level sortition panel organized by the Sortition Foundation in 2020.
The panel consisted of $k=110$ many UK residents aged 16 and above.
The Sortition Foundation invited all members of 30\,000 randomly selected households, which reached an estimated $r = 60\,000$ eligible participants.\footnote{Note that every person in the population has equal probability $(30\,000/\text{\#households})$ of being invited. We ignore correlations between members of the same household.}
Of these letter recipients, 1\,715
participated in the pool,\footnote{Excluding 12 participants with gender ``other'' as no equivalent value is present in the background data.} corresponding to a mean participation probability of $\overline{q} \approx 2.9\%$.
The feature-value pairs used for this panel can be read off the axis of \cref{fig:realized_representation}.
We omit an additional feature \textit{climate concern level} in our main analysis because only 4 members of the pool have the value \textit{not at all concerned}, whereas this feature-value pair's proportional number of panel seats is 6.5.
To allow for proportional representation of groups with such low participation rates, $r$ should have been chosen to be much larger.
We believe that the merits of our algorithm can be better observed in parameter ranges in which proportionality can be achieved.
For the background sample, we used the 2016 European Social Survey~\cite{ess}, %
which contains 1\,915 eligible individuals, all with features and values matching those from the panel.
Our implementation is based on PyTorch and Gurobi, runs on consumer hardware, and its code is available on \href{https://github.com/pgoelz/endtoend}{github}.
\Cref{app:empirical} contains details on Climate Assembly UK, data processing, the implementation, and further experiments (including the climate concern feature).

\noindent\textbf{Estimation of $\bm{\beta}$ Parameters.}
We find that the baseline probability of participation is $\beta_0 = 8.8\%$.
Our $\beta_{f,v}$ estimates suggest that (from strongest to weakest effect) highly educated, older, urban, male, and non-white agents participate at higher rates. These trends reflect these groups' respective levels of representation in the pool compared to the underlying population, suggesting that our estimated $\beta$ values fit our data well.
Different values of the remaining feature, region of residence, seem to have heterogeneous effects on participation, where being a resident of the South West gives substantially increased likelihood of participation compared to other areas.
The lowest participation probability of any agent in the pool, according to these estimates, is $q^* = 0.78\%$, implying that $\alpha \approx 4.25$. See \cref{sec:app_betavalidation} for detailed estimation results and validation.

\begin{figure}[tbp]
    \centering
    \hspace*{0.15cm}\includegraphics[width=0.8\textwidth]{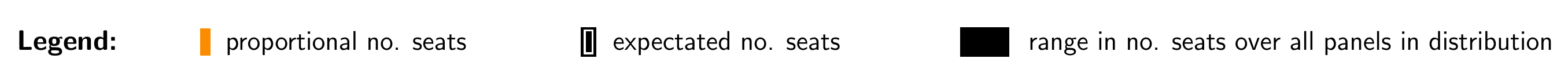}\\[-.5em]
    \includegraphics[width=\textwidth]{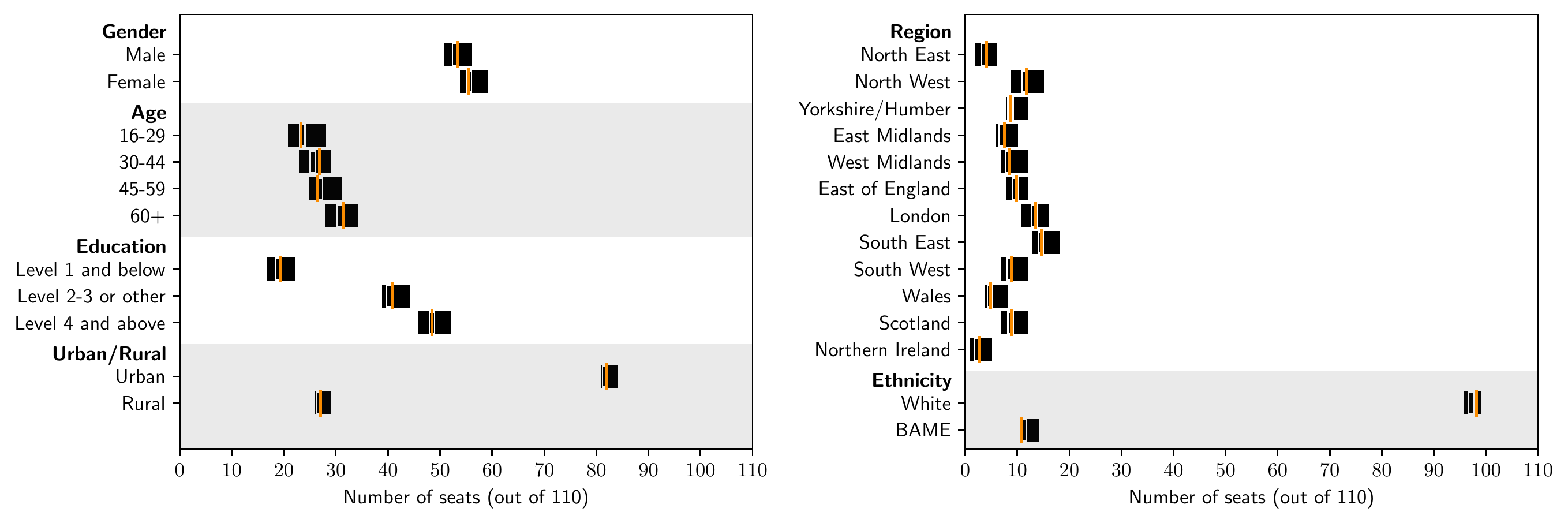}
    \caption{Expected and realized numbers of panel seats our algorithm gives each feature-value pair in the Climate Assembly pool.}%
    \label{fig:realized_representation}
\end{figure}
\noindent\textbf{Running the Sampling Algorithm on the Pool.} The estimated $q_i$ allow us to run our algorithm on the Climate Assembly pool and thereby study its fairness properties for non-asymptotic input sizes.
We find that the Climate Assembly pool is good relative to our $q_i$ estimates, i.e., that it satisfies \cref{eq:good1,eq:good2,eq:good3}.
As displayed in \cref{fig:realized_representation}, the marginals produced by Phase~I of our algorithm give each feature-value pair $f,v$ an expected number of seats, $\sum_{i \in P,f(i) = v} \pi_{i,P}$, within \textit{one seat} of its proportional share of the panel, $k \, n_{f,v}/n$.
By \cref{lem:beckfiala_random}, Phase~II of our algorithm then may produce panels from these marginals in which $f,v$ receives up to $|F|=6$ fewer or more seats than its expected number.
However, as the black bars in \cref{fig:realized_representation} show, the actual number of seats received by any $f,v$ across \emph{any panel} produced by our algorithm on this input never deviates from its expectation by more than 4 seats.
As a result, while \cref{thm:main} only implies lower quotas of $.51 \, k \, n_{f,v} / n - |F|$ and upper quotas of $1.49 \, k \, n_{f,v} / n + |F|$ for this instance, the shares of seats our algorithm produces lie in the much narrower range $k \, n_{f,v}/n \pm 5$ (and even $k \, n_{f,v}/n \pm 3$ for 18 out of 25 feature-value pairs).
This suggests that, while the quotas guaranteed by our theoretical results are looser than the quotas typically set by practitioners, our algorithm will often produce substantially better ex-post representation than required by the quotas.

\looseness=-1
\noindent\textbf{End-to-End Probabilities.}
In the previous experiments, we were only able to argue about the algorithm's behavior on a single pool.
To validate our guarantees on individual end-to-end probabilities, we construct a synthetic population of size 60 million by duplicating the ESS participants, assuming our estimated $q_i$ as their true participation probabilities.
Then, for various values of $r$, we sample a large number of pools. %
By computing $\pi_{i,P}$ values for all agents $i$ in each pool, we can estimate each agent's end-to-end probability of ending up on the panel.
Crucially, we assume that our algorithm does not produce any panel for bad pools, analogously to \cref{thm:main}.
As shown in the following graph, for $r = 60\,000$ (as was used in Climate Assembly UK), all agents in our synthetic population, across the full range of $q_i$, receive probability within $.1\,k/n$ of $k/n$ (averaged over 100\,000 random pools): \\[7pt]
\includegraphics[width=\textwidth]{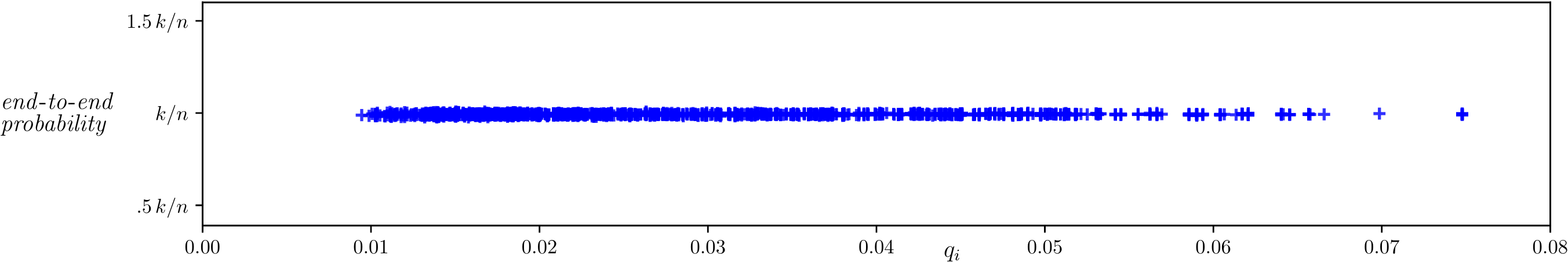}
That these end-to-end probabilities are so close to $k/n$ also implies that bad pools are exceedingly rare for this value of $r$.
As we show in \cref{sec:app_experiments}, we see essentially the same behavior for values of $r$ down to roughly $15\,000$, when $\alpha \approx 1$.
For even lower $r$, most pools are bad, so end-to-end probabilities are close to zero under our premise that no panels are produced from bad pools.

To demonstrate that our algorithm's theoretical guarantees lead to realized improvements in individual fairness over the state-of-the-art, we re-run the experiment above, this time using the Sortition Foundation's greedy algorithm to select a panel from each generated pool. 
Since their algorithm requires explicit quotas as input, we set the lower and upper quotas for each feature-value group to be the floor and ceiling of that group's proportional share of seats. This is a popular way of setting quotas in current practice.

\includegraphics[width=\textwidth]{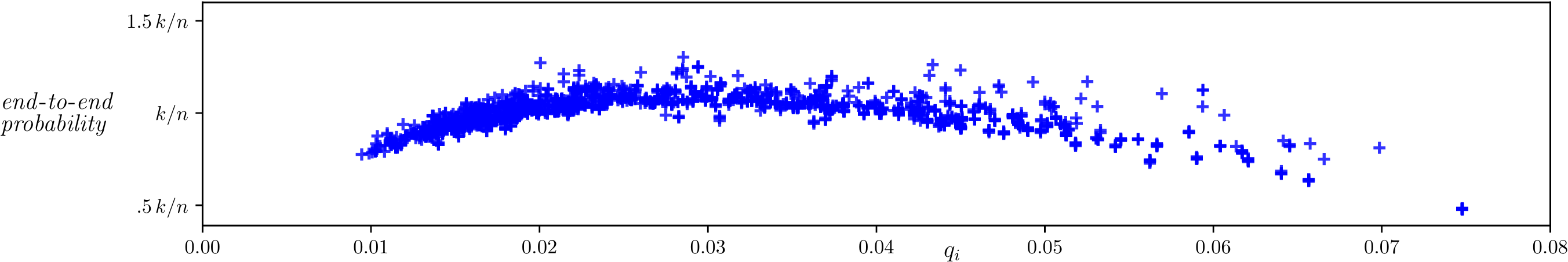}

The results of this experiment show that the individual end-to-end probabilities generated by the currently-used greedy algorithm range from
below $0.5 \, k/n$ up to $1.3 \, k/n$.
In comparison to %
the end-to-end probabilities generated by our algorithm,
those generated by the greedy algorithm are substantially skewed, and tend to disadvantage individuals with either low or high participation probabilities. 
One might argue that the comparison between our algorithm and the greedy is not quite fair, since the greedy algorithm is required to satisfy stronger quotas.
However, looser quotas do not improve the behavior of the greedy algorithm; they simply make it behave more similarly to uniform sampling from the pool, which further disadvantages agents with low participation probability (%
for details, see \cref{app:endtoenddetails}).

Taken together, these results illustrate that, although greedy algorithms like the one we examined achieve proportional representation of a few pre-specified groups via quotas, they do not achieve fairness to individuals or to groups unprotected by quotas.
Compared to the naive solution of uniform sampling from the pool, greedily striving for quota satisfaction does lead to more equal end-to-end probabilities, as pool members with underrepresented features are more likely to be selected for the panel than pool members with overrepresented features.
However, this effect does not neutralize self-selection bias when there are multiple features, even when selection bias acts through the independent-action model as in our simulated population.
Indeed, in this experiment, the greedy algorithm insufficiently boosts the probabilities of agents in the intersection of multiple low-participation groups (the agents with lowest $q_i$), while also too heavily dampening the selection probability of those in the intersection of multiple high-participation groups (with highest $q_i$). These observations illustrate the need for panel selection algorithms that explicitly control individual probabilities.

\section{Discussion} \label{sec:discussion}
In a model in which agents \emph{stochastically} decide whether to participate, our algorithm guarantees similar end-to-end probabilities to all members of the population.
Arguably, an agent's decision to participate when invited might not be random, but rather \emph{deterministically} predetermined.

From the point of view of such an agent $i$, does our algorithm, based on a model that doesn't accurately describe her (and her peers') behavior, still grant her individual fairness?
If $i$ deterministically \emph{participates}, the answer is yes (if not, of course she cannot be guaranteed anything).
To see why, first observe that, insofar as it concerns $i$'s chance of ending up on the panel, all other agents might as well participate randomly.\footnote{Fix a group of agents who, assuming the stochastic model, will participate if invited with probability $q$. Then, sampling letter recipients from this set of agents in the stochastic model is practically equivalent to sampling recipients from this group in the deterministic model, if a $q$ fraction of the group deterministically participate.}
Indeed, from agent~$i$'s perspective, the process looks like the stochastic process where every other agent $j$ participates with probability $q_j$, where $i$ herself always participates, and where the algorithm erroneously assumes that $i$ joins only with some probability $q_i$.
Therefore, the pool is still good with high probability conditioned on $i$ being in it, as argued in \cref{lem:goodpools}. Even if the algorithm knew that $q_i=1$, $i$'s end-to-end probability would be at least $\big(1 - o(1)\big) \, k/n$, and the fact that the algorithm underestimates her $q_i$ only increases her probability of being selected from the pool. It follows that $i$'s end-to-end probability in this setting still must be at least around $k/n$.

Thus, in a deterministic model of participation, our individual guarantees are reminiscent of the axiom of population monotonicity in fair division:
\emph{If the whole population always participated when invited, every agent would reach the panel with probability $k/n$. The fact that some agents do not participate cannot (up to lower-order terms) decrease the selection probabilities for those who do.}

\section*{Broader Impact}

As we discussed in the paper, sortition is becoming increasingly widespread as a method for making collective decisions and gauging public opinion. Based on our experiences, practitioners seem interested in the possibility of fairer sampling algorithms, so our research in this area has the potential to influence how sortition panels are sampled in the real world.

From a \textbf{fairness} standpoint, our algorithm represents an improvement over currently-used algorithms on several fronts. We maintain the approximate satisfaction of proportional quotas while giving provable fairness guarantees to individuals, which in turn safeguard against systematic under-representation of demographic groups unprotected by quotas. Currently-used greedy algorithms give no such guarantees, and as we show in a real-world example, these theoretical differences can translate to substantial practical differences in the equality of individual selection probabilities. Our model is also robust to bias in the first stage of sampling: if a group is unfairly undersampled in the invitation stage, an effective machine-learning method will estimate lower $q_i$ for this group, and the algorithm will then correct this bias in the second stage. For these reasons, sortition panels selected via our algorithm will in general be more representative of the population, hopefully resulting in their decisions more fairly reflecting the interests and views of the entire population.

Of course, the methods we present do still come with fairness concerns. One main concern would be that the learned $q_i$ values could erroneously overestimate the participation probabilities of those in a certain group, thereby resulting in our sampling algorithm giving members of that group lower probability they deserve in the second step, and thereby giving them unfairly low end-to-end selection probability. We identify three main potential sources of inaccuracy in our MLE estimator: unreliable or unrepresentative background data, small sample sizes from which to learn, and failure of our model to accurately capture people's participation decisions. While each of these factors could potentially introduce inaccuracies in the $q_i$ values, it is unclear that any but the first issue is likely to result in systematic bias against certain groups.

Another important consideration in panel selection is \textbf{transparency} about the selection process to constituents on whose behalf the panel will make decisions. Increasing the transparency of how these panels are selected can increase public trust in the fairness of this method of decision-making, thus expanding sortition's legitimacy and reach. Toward the goal of transparency, our approach offers improvements on the state-of-the-art in several ways. First, this work gives formal theoretical guarantees of multiple types of fairness, where currently-used methods give none. Our algorithm also follows an explicit and explainable process for setting marginal selection probabilities, rather than having them accidentally arise from a greedy process.

Despite these improvements, our methods still present transparency challenges.
Since an individual's probability of selection from the pool depends on their estimated $q_i$ value, the fairness of the process hinges on the entire machine-learning pipeline\,---\,data used, choice of model, and estimation methods\,---\,multiple elements of which might be opaque to most of the population. Secondly, while the use of protected attributes to counteract inequality is an accepted practice in the fairness in classification literature, there may still be public discomfort about this idea. In particular, it may be a source of discomfort that, in order to equalize end-to-end probabilities, our algorithm must explicitly decrease the probability of certain people being selected from the pool due to their attributes. One potentially comforting and easily-understood feature of our algorithm, however, is described in \Cref{sec:discussion}: that if someone participates in the pool with probability 1, they are guaranteed an end-to-end selection probability of at least $k/n$, regardless of their attributes.

\begin{ack}
We thank Sivaraman Balakrishnan, Nikhil Bansal, and David Wajc for helpful technical discussions, and Terry Bouricious, Adam Cronkright, Linn Davis, Adela G\k{a}siorowska, Marcin Gerwin, Brett Hennig, David Schecter, and Robin Teater for sharing their insights on practical sortition.
We would also like to express our gratitude to the Sortition Foundation for supplying the data used in our experiments. 

This work was partially supported by the National Science Foundation under grants CCF-1907820, CCF1955785, CCF-2006953, CCF-1525932, CCF-1733556, CCF-2007080, and IIS-2024287; and by the Office of Naval Research under grant N00014-20-1-2488.
Bailey Flanigan is supported by the National Science Foundation Graduate Research Fellowship and the Fannie and John Hertz Foundation.
\end{ack}

\bibliographystyle{plainurl}
\bibliography{bibliography}

\newpage
\appendix
\section*{\Large Appendix}
\section{Notation Glossary}
\begin{center}
        \begin{tabular}{ll}
    \toprule
    \multicolumn{2}{l}{\textbf{Sets of Agents}}\\
    $N$ & Set of agents in the population\\
    $\Recipients$ & Set of agents who receive invitation letters (random variable)\\
    $\Pool$ & Set of agents in the pool (random variable)\\
    $\Panel$ & Set of agents on the panel (random variable)\\
    \midrule
    \multicolumn{2}{l}{\textbf{Sortition Panel Parameters}}\\
     $n$ & Size of the population \\
     $r$ & Number of invitation letters sent out \\
     $k$ & Size of the panel\\
     $F$ & Set of all features\\
     $V_f$ & Set of possible values for a specific feature $f \in F$\\
     $F(i)$ & Feature vector of agent $i$\\
     $n_{f,v}$  & Number of agents in the population with value $v$ of feature $f$\\
     $\ell_{f,v}, u_{f,v}$  & Lower and upper quotas for every feature-value pair\\
     $q_i$  & Probability that agent $i \in N$ enters the pool, conditioned on being invited\\
     $q^*$ & Minimum value of $q_i$ over all agents ($q^* := \min_{i \in N} q_i$)\\
     $\alpha$ & Parameter defined as $\alpha := q^*\,r/k$\\
     \bottomrule\\
\end{tabular}
\end{center}

\section{Supplementary Material for \cref{sec:theory}} \label{sec:app_theory}
\subsection{Discussion of Theorem Preconditions}
\label{sec:preconditions}

We show that pools are good with high probability under two preconditions: that each feature-value group constitutes at least $1/k$ fraction of the population (so $n_{f,v}/n \geq 1/k$ for all $f,v$), and that the number of recipients is sufficiently high relative to the participation probabilities and the panel size ($\alpha = q^* \, r / k \to \infty$).

The first condition is natural because if a group should proportionally receive less than one seat on the panel, any positive lower bound on selection probabilities for agents in groups would violate proportionality.

The second condition enforces that the number of agents invited $r$ is large enough relative to the minimum participation probability $q^*$ and the size of the panel. Without this condition, there can be a constant probability that the pool will feature zero agents with a certain feature-value: Suppose that $\alpha$ is an arbitrary positive constant, set all $q_i \coloneqq \alpha \, k / r$, and consider a feature-value pair $f,v$ with $n_{f,v} = n/k$ agents. 
In expectation, there will be $(r / n) \, (n / k) = r / k$ agents with feature-value $f,v$ among the recipients.
If $r \in \omega(k)$, there are at most $2 \, r / k$ such recipients with high probability. 
Then, the probability that the pool contains no agent with $f,v$ is at least 
\[(1 - \alpha \, k / r)^{2 \, r/k} = (1 - q_i)^{2 \, \alpha / q_i} = \big( \underbrace{(1-q_i)^{1/q_i}}_{\smash{\text{$\to 1/e$ as $q_i \to 0$}}}\big)^{2 \, \alpha} \to e^{-2 \, \alpha} > 0.\]

\subsection{Discussion of Ties to Discrepancy Theory}

In rounding agents' marginal selection probabilities to select a panel, we round fractional variables to 0 or 1 such that the sum of certain sets of variables changed only by a small amount.
This problem is closely connected to \emph{combinatorial discrepancy}~\cite{Chazelle01,Spencer85}, which can be summarized in the same words, by additionally assuming that the initial fractional values are $1/2$.
In fact, the original Beck-Fiala theorem arises in the context of discrepancy, showing that, if each variable appears in a bounded number $t$ of sets, discrepancy $\Theta(t)$ can be achieved (where in our setting, $t$ corresponds to $|F|$, the number of features).
Beck and Fiala~\cite{BF81} conjectured that it is actually possible to achieve discrepancy in $\mathcal{O}(\sqrt{t})$.
Should this conjecture be true, similar ideas might translate to our setting to guarantee the satisfaction of quotas closer to exact proportionality.
To this day, however, the best known bound in $t$ is still in $\Theta(t)$~\cite{Bukh16}. In accordance with this result, we guarantee a relaxation of $|F|$ from proportional representation of groups.

We note that there do exist other disrepancy results that give sub-linear dependencies on $|F|$, but at the cost of introducing dependencies on other parameters. One such result is Theorem 5.3 in \cite{bansal2019}, which guarantees discrepancy a square-root dependency on $|F|$. However, subject to our requirement that the per-person marginal probability must deviate from $k/n$ by only $\pm \delta k/n$ where $\delta \in o(1)$, Bansal's result guarantees a discrepancy bound of $O(\sqrt{|F| \log(k n/\delta)})$, which grows in $n$, making it unfavorable in our setting.

\subsection{Proof of \cref{lem:goodpools}}
\label{sec:app_goodpools}
The results in this section allow $k \geq 1$ and $r \geq 1$ to vary arbitrarily in $n$; they just require that $\alpha := q^* r / k \to \infty$ as $n \to \infty$ (without requiring $\alpha$ to grow at a specific minimum rate relative to $n$).
All convergences are relative to $n$ going to infinity.
\begin{customlemma}{2}
Suppose that $\alpha \to \infty$ and $n_{f,v} \geq n/k$ for all $f,v$.
Then, for all agents $i \in \Population$, $\mathbb{P}[\text{$\Pool$ \emph{is good}} \mid i \in \Pool] \to 1$.
\end{customlemma}

In the following proofs, it is convenient to refer to $1/q^*$, the largest possible value of $a_i$, as $a^*$. Note that $a^* = \frac{r}{\alpha \, k}$. We will refer to the random set of recipients with a certain feature-value pair $f,v$ as $\Recipients_{f,v} \coloneqq \{i \in \Recipients \mid f(i)=v\}$.

We begin by showing in \cref{lem:conditionalab,lem:conditionalc} that, conditioned on $i$ being in the pool, the following three events occur with high probability:
\begin{enumerate}[label=\textbf{\Alph*.}]
    \itemsep0em
    \item $k \, a^* \leq \sum_{j \in \Pool} a_{j}$
    \item $\sum_{j \in \Pool} a_j \in [(1 - \alpha^{-.492}) \, r, (1 + \alpha^{-.492}) \, r]$
    \item $\sum_{j \in \Pool : f(j) = v} a_j \in [(1 - \alpha^{-.492}) \, \frac{n_{f, v}}{n} \, r, (1 + \alpha^{-.492}) \, \frac{n_{f,v}}{n} \, r] \quad \forall f,v$
\end{enumerate}
We then show in \cref{lem:abcimplygood} that, when these events occur on some pool, the pool must be good, which concludes the proof of \cref{lem:goodpools}.

\begin{lemma}
\label{lem:conditionalab}
Under the assumptions of \cref{lem:goodpools}, $\mathbb{P}\left[\textbf{Event A} \land \textbf{Event B} \,\middle|\, i \in \Pool \right] \to 1$.
\end{lemma}
\begin{proof}
    Fix the set of recipients $R$ (including $i$).
With respect to the randomness in the pool self-selection, the random variables $a_{j} \cdot \mathds{1}\{j \in \Pool\}$ across all $j \in R \setminus \{i\}$ are independent, bounded in $[0, a^*]$, and have expected value $a_j \, q_j = 1$.
Thus, by a Chernoff bound, and using that $a^* = r/(\alpha \, k)$,
\begin{align*}\mathbb{P}\left[\left|\sum_{j \in \Pool \setminus \{i\}} a_{j} - (r - 1)\right| \geq \alpha^{-.495} \, (r - 1)\right] &\leq 2 \, e^{- \alpha^{-.99} \, \frac{r-1}{a^*} / 3} \\
&= 2 \, e^{- \alpha^{-.99} \, \frac{r-1}{r} \, \alpha \, k / 3} \\
&\leq 2 \, e^{- \Omega(\alpha^{.01})} \to 0,
\end{align*}
where the last inequality uses the fact that $r \geq 2$ for large enough $n$\footnote{Since $r = \alpha \, k / q^* \geq \alpha / q^* \geq \alpha \to \infty$.} and that $k \geq 1$.

Conditioning on this high-probability event, it follows that, for large enough $n$,
\[ \sum_{j \in \Pool} a_{j} \geq 1 + \sum_{j \in \Pool \setminus \{i\}} a_{j} \geq 1 + (1 - \alpha^{-.495}) \, (r - 1) \geq (1 - \alpha^{-.492}) \, r, \]
which shows the lower bound in Event~B.
For the upper bound,
\begin{align*}
    \sum_{j \in \Pool} a_j &\leq a^* + \sum_{j \in \Pool \setminus \{i\}} a_j \leq a^* + (1 + \alpha^{-.495}) \, (r - 1) \leq r/(\alpha \, k) + (1 + \alpha^{-.495}) \, r \\
    &\leq (1 + \alpha^{-.495} + 1/\alpha) \, r \leq (1 + \alpha^{-.492}) \, r \leq 1/(1 - \alpha^{-.492}) \, r.
\end{align*}
This establishes Event~B.

For large enough $n$, the lower bound on $\sum_{j \in \Pool} a_{j}$ can be extended as
\[\sum_{j \in \Pool} a_{j}  \geq (1 - \alpha^{-.492}) \, r \geq r/\alpha \geq k \, a^*,\]
which shows Event~A.
\end{proof}

For Event~C, we need to show that $\sum_{j\in \Pool : f(j)=v} a_j$ is concentrated for a feature-value pair $f,v$.
As an intermediate step, we first show that the \emph{number} of pool members (``$\sum_{j\in \Pool : f(j)=v} 1$'') with this feature-value pair is concentrated: 
\begin{lemma} \label{lemma:concentration_rfv}
    Under the assumptions of \cref{lem:goodpools}, for each $f,v$,
    \[\mathbb{P}\left[(1 - \alpha^{-.495}) \, \frac{n_{f,v}}{n} \, r \leq \left|\Recipients_{f,v}\right| \leq (1 + \alpha^{-.495}) \, \frac{n_{f,v}}{n} \, r \,\middle|\, i \in \Pool\right] \to 1.\]
\end{lemma}
\begin{proof}
    Conditioned on $i \in \Pool \subseteq \Recipients$, $\Recipients \setminus \{i\}$ is distributed as if $r - 1$ members of $\Population \setminus \{i\}$ were drawn with equal probability and without replacement.
    Thus, 
    \[ \mathbb{E}\left[\left|\Recipients_{f,v}\right| \,\middle|\, i \in \Pool\right] = \begin{cases} n_{f,v} \, \frac{r - 1}{n - 1} & \text{if $f(i) \neq v$} \\
    1 + (n_{f,v} - 1) \, \frac{r-1}{n-1} & \text{if $f(i) = v$.}\end{cases}\]
    In both cases, we show that $\mathbb{E}\left[\left|\Recipients_{f,v}\right| \,\middle|\, i \in \Pool\right] \in [(1 - k/r) \, n_{f,v} \, \frac{r}{n}, (1 + k/r) \, n_{f,v} \, \frac{r}{n}]$.
    Indeed, for the upper bound,
    \begin{align*}
        \mathbb{E}\left[\left|\Recipients_{f,v}\right| \,\middle|\, i \in \Pool\right] &\leq 1 + (n_{f,v} - 1) \, \frac{r-1}{n-1} \leq 1 + n_{f,v} \, \frac{r}{n} = \left(1 + \frac{n}{n_{f,v}} / r\right) \, n_{f,v} \, \frac{r}{n} \\
        &\leq (1 + k/r) \, n_{f,v} \, \frac{r}{n} \leq (1 + 1/\alpha) \, n_{f,v} \, \frac{r}{n}.
    \end{align*}
    For the lower bound,
    \begin{align*}
        \mathbb{E}\left[\left|\Recipients_{f,v}\right| \,\middle|\, i \in \Pool\right] &\geq n_{f,v} \, \frac{r - 1}{n - 1} = \frac{r-1}{r} \, n_{f,v} \, \frac{r}{n} = (1 - 1/r) \, n_{f,v} \, \frac{r}{n} \\
        &\geq (1 - k/r) \, n_{f,v} \, \frac{r}{n} \geq (1 - 1/\alpha) \, n_{f,v} \, \frac{r}{n}.
    \end{align*}

    As the (independent) union of the deterministic set $\{i\}$ and indicator variables for sampling without replacement, the variables $\mathds{1}\{j \in \Recipients\}$ satisfy negative association and therefore Chernoff inequalities~\cite{Wajc17}.
    Thus, for the upper tail bound,
\begin{align*}
    &\mathbb{P}\left[\left|\Recipients_{f,v}\right| \geq (1 + \alpha^{-.497}) \, (1 + 1/\alpha) \, n_{f,v} \, \frac{r}{n} \middle| \, i \in \Pool \right] \leq e^{- \alpha^{-.994} \, (1 + 1/\alpha) \, n_{f,v} \, \frac{r}{n} / 3} \\
    &\leq e^{- \alpha^{-.994} \,  n_{f,v} \, \frac{r}{n} / 3} 
\leq e^{- \alpha^{-.994} \, \frac{r}{k} / 3} \leq e^{- \alpha^{-.994} \, \alpha / 3} \leq e^{- \alpha^{.006} / 3} \to 0.
\end{align*}
Similarly, for the lower tail bound,
\begin{align*}
    &\mathbb{P}\left[\left|\Recipients_{f,v}\right| \leq (1 - \alpha^{-.497}) \, (1 - 1/\alpha) \, n_{f,v} \, \frac{r}{n} \middle| \, i \in \Pool \right] \leq e^{- \alpha^{-.994} \, (1 - 1/\alpha) \, n_{f,v} \, \frac{r}{n} / 2} \\
    &\stackrel{\mathclap{(\alpha \geq 3)}}{\leq} \;\; e^{- \alpha^{-.994} \, n_{f,v} \, \frac{r}{n} / 3} 
        \leq e^{- \alpha^{.006} / 3} \to 0.
\end{align*}

The claim follows from observing that, for $r/k$ large enough,
\[ (1 - \alpha^{-.497}) \, (1 - 1/\alpha) \geq 1 - \alpha^{-.497} - \alpha^{-1} \geq 1 - \alpha^{-.495}\]
and
\[ (1 + \alpha^{-.497}) \, (1 + 1/\alpha) = 1 + \alpha^{-.497} + \alpha^{-1} + \alpha^{-1.497} \leq 1 - \alpha^{-.495}. \qedhere\]
\end{proof}

\begin{lemma}
\label{lem:conditionalc}
Under the assumptions of \cref{lem:goodpools}, $\mathbb{P}\left[\textbf{Event C} \,\middle|\, i \in \Pool \right] \to 1$.
\end{lemma}
\begin{proof}
    Fix a single feature-value pair $f,v$.
    By \cref{lemma:concentration_rfv}, with high probability, the number of recipients $r_{f,v}$ with feature-value pair $f,v$ is in
    \[\left[(1 - \alpha^{-.495}) \, \frac{n_{f,v}}{n} \, r, (1 + \alpha^{-.495}) \, \frac{n_{f,v}}{n} \, r\right].\]
    Going forward, we will fix a set of recipients $R$, and we assume that $r_{f,v}$ indeed falls in this range.
    For large enough $n$, this implies that $r_{f,v}$ is positive.
    For ease of notation, we will implicitly condition on $i \in \Pool$ and these high-probability events.

    The self-selection process of agents with feature-value pair $f, v$ might look a bit different depending on whether $f(i)=v$.
    If $f(i) \neq v$, the self selection of agents with feature-value pair $f,v$ is independent from our knowledge about $i$ being in the pool.
    Thus, the random variable $\sum_{\substack{j \in \Pool,\\f(j) = v}} a_j$ is the sum of independent random variables $a_j \, \mathds{1}\{j \in \Pool\}$ for each $j \in R, f(j)=v$, where each variable is bounded in $[0, a^*]$ and has expectation $1$.
    In particular, $\mathbb{E}\left[\sum_{\substack{j \in \Pool,\\f(j) = v}} a_j\right] = r_{f,v}$.

    Else, if $f(i) \neq v$, $\sum_{\substack{j \in \Pool,\\f(j) = v}} a_j$ is still the sum of independent random variables $a_j \, \mathds{1}\{j \in \Pool\}$ and each variable is bounded in $[0, a^*]$. However, the specific variable $a_i \, \mathds{1}\{i \in \Pool\}$ is deterministically $a_i$ (all other variables still have expectation $1$).
    Thus, $\mathbb{E}\left[\sum_{\substack{j \in \Pool,\\f(j) = v}} a_j\right] = r_{f,v} - 1 + a_i$.
    \begin{align*}
    r_{f,v} - 1 + a_i &= \left(1 + \frac{a_i - 1}{r_{f,v}}\right) \, r_{f,v} \leq \left(1 + \frac{a^*}{r_{f,v}}\right) \, r_{f,v} \leq \left(1 + \frac{r/(\alpha \, k)}{(1 - \alpha^{-.495}) \, r \, n_{f,v}/n}\right) \, r_{f,v} \\
    &\leq \left(1 + \frac{r/(\alpha \, k)}{(1 - \alpha^{-.495}) \, r / k}\right) \, r_{f,v} = \left(1 + \frac{1}{(1 - \alpha^{-.495}) \, \alpha}\right) \, r_{f,v} \\
    &\leq (1 + 2/\alpha) \, r_{f,v}. \tag{for $\alpha^{.495} \geq 2$}
\end{align*}
Thus, across both cases, the expectation  $\mathbb{E}\left[\sum_{\substack{j \in \Pool,\\f(j) = v}} a_j\right]$ is at least $r_{f,v} \geq (1 - \alpha^{-.495}) \, \frac{n_{f,v}}{n} \, r$ and at most $(1 + 2 / \alpha) \, r_{f,v} \leq (1 + 2/\alpha) \, (1 + \alpha^{-.495}) \, \frac{n_{f,v}}{n} \, r \leq (1 + \alpha^{-.493})\, \frac{n_{f,v}}{n} \, r$ for large $n$, and we can use Chernoff bounds.
    
For bounding the lower tail,
\begin{align*}
    &\mathbb{P}\left[ \sum_{\substack{j \in \Pool,\\f(j) = v}} a_j \leq (1 - \alpha^{-.495}) \, (1 - \alpha^{-.495}) \, \frac{n_{f,v}}{n} \, r\right] \leq e^{- \alpha^{-.99} \, (1 - \alpha^{-.495}) \, \frac{n_{f,v}}{n} \, r / (2 \, a^*)} \\
    &\stackrel{\mathclap{(\alpha^{.495} \geq 3)}}{\leq} \;\;\;\; e^{- \alpha^{-.99} \, \frac{n_{f,v}}{n} \, r / (3 \, a^*)} 
    = e^{- \alpha^{-.99} \, \frac{n_{f,v}}{n} \, r / (3 \, r / (\alpha \, k))} 
    \leq e^{- \alpha^{-.99} \, \frac{n_{f,v}}{n} \, \alpha \, k / 3} \\
    &\leq e^{- \alpha^{-.99} \, \alpha / 3}  \\
    &\leq e^{- \alpha^{.01} / 3}  \to 0.
\end{align*}
For bounding the upper tail,
\begin{align*}
    &\mathbb{P}\left[ \sum_{\substack{j \in \Pool,\\f(j) = v}} a_j \geq (1 + \alpha^{-.495}) \, (1 + \alpha^{-.493}) \, \frac{n_{f,v}}{n} \, r\right] \leq e^{- \alpha^{-.99} \, (1 + \alpha^{-.493}) \, \frac{n_{f,v}}{n} \, r / (3 \, a^*)} \\
    &\leq e^{- \alpha^{-.99} \, \frac{n_{f,v}}{n} \, r / (3 \, a^*)} 
    = e^{- \alpha^{-.99} \, \frac{n_{f,v}}{n} \, \alpha \, k / 3} 
    \leq e^{- \alpha^{-.99} \, \alpha / 3} 
    \leq e^{- \alpha^{.01} / 3} \to 0.
\end{align*}

Note that, for large $n$,
$(1 - \alpha^{-.495}) \, (1 - \alpha^{-.495}) \geq 1 - 2 \, \alpha^{-.495} \geq 1 - \alpha^{-.492}$.
Similarly, $(1 + \alpha^{-.495}) \, (1 + \alpha^{-.493}) \in 1 + \mathcal{O}(\alpha^{-.493})) \leq 1 + \alpha^{-.492}$.

This shows that, for each $f,v$, $(1 - \alpha^{-.492}) \, \frac{n_{f,v}}{n} \, r \leq \sum_{\substack{j \in \Pool,\\f(j) = v}} a_j \leq (1 + \alpha^{-.492}) \, \frac{n_{f,v}}{n} \, r$ with high probability.
The claim follows by a union bound over all (finitely many) feature-value pairs.
\end{proof}

\begin{lemma}
\label{lem:abcimplygood}
For large enough $n$, if Events A, B, and C occur for a pool $P$, $P$ is good.
\end{lemma}
\begin{proof}
Suppose that Events~A, B, and C occur in a pool $P$.

\paragraph{Condition~\eqref{eq:good1}: $\forall j \in P. \; 0 \leq \pi_{j, P} \leq 1$.}
Clearly, $\pi_{j, P}$ is nonnegative, and Event~A implies that $\pi_{j,P} = k \, a_j / \sum_{j' \in P} a_{j'} \leq k \, a^* / \sum_{j' \in P} a_{j'} \leq 1$.

\paragraph{Condition~\eqref{eq:good2}: $\forall f, v.\;(1 - \alpha^{-.49})\, k \, n_{f, v} / n \leq \sum_{j \in P : f(j) = v} \pi_{j, P} \leq (1 + \alpha^{-.49}) \, k \, n_{f, v} / n$.}
Fix any feature-value pair $f,v$.
Recall that, by Event~B, 
\[\sum_{j \in P} a_j \in [(1 - \alpha^{-.492}) \, r, (1 + \alpha^{-.492}) \, r],\]
and, by Event~C, 
\[\sum_{j \in P : f(j) = v} a_j \in [(1 - \alpha^{-.492}) \, \frac{n_{f, v}}{n} \, r, (1 + \alpha^{-.492}) \, \frac{n_{f,v}}{n} \, r].\]

Observe that, for any $x \in [0, 1/3]$,
\[\frac{1 + x}{1-x} \leq \frac{1 + x + x \, (1 - 3 \, x)}{1-x} = \frac{1 + 2 \, x - 3 \, x^2}{1 - x} = 1 + 3\,x. \]
Then, if $n$ is large enough such that $\alpha^{-.492} \leq 1/3$, it follows that
\begin{align*}
\sum_{j \in P : f(j) = v} \pi_{j,P} &= k \, \frac{\sum_{j \in P : f(j)=v} a_j}{\sum_{j \in P} a_j} \leq k \, \frac{(1 + \alpha^{-.492}) \, \frac{n_{f,v}}{n} \, r}{(1 - \alpha^{-.492}) \, r} \leq (1 + 3 \, \alpha^{-.492}) \, k \, \frac{n_{f,v}}{n} \\
&\leq (1 + \alpha^{-.49}) \, k \, \frac{n_{f,v}}{n}.
\end{align*}

Next, observe that, for any $x$,
\[ \frac{1 - x}{1 + x} \geq \frac{1 - x - 2 \, x^2}{1 + x} = 1 - 2 \, x. \]
Thus,
\begin{align*}
\sum_{j \in P : f(j) = v} \pi_{j,P} &= k \, \frac{\sum_{j \in P : f(j)=v} a_j}{\sum_{j \in P} a_j} \geq k \, \frac{(1 - \alpha^{-.492}) \, \frac{n_{f,v}}{n} \, r}{(1 + \alpha^{-.492}) \, r} \geq (1 - 2 \, \alpha^{-.492}) \, k \, \frac{n_{f,v}}{n} \\
&\geq (1 - \alpha^{-.49}) \, k \, \frac{n_{f,v}}{n}.
\end{align*}

\paragraph{Condition~\eqref{eq:good3}: $\sum_{i \in P}a_i \leq r/(1-\alpha^{-.49})$.}
This follows from Event~B since $\sum_{j \in P} a_j \leq (1 + \alpha^{-.492}) \, r \leq (1 + \alpha^{-.49}) \, r = \frac{1 - \alpha^{-.98}}{1 - \alpha^{-.49}} \, r \leq r/(1 - \alpha^{-.49})$ for large enough $n$.
\end{proof}

\subsection{Proof of \cref{lem:beckfiala_random}}
\label{sec:app_bf}
\subsubsection{Rounding the Linear Program Using Discrepancy Methods}
In Part~II of the algorithm, we need to implement the marginal probabilities $\pi_{i, P}$ from Part~I by randomizing over panels of size $k$.
Additionally, the panels produced by this procedure should guarantee that the number of panel members of a feature-value pair $(f,v)$ lies in a narrow interval around the proportional number of panel members $k \, n_{f,v} / n$.
Technically, this corresponds to randomly rounding the fractional solution $x_i \coloneqq \pi_{i, P}$ of an LP, %
such that afterwards all variables are 0 or 1, i.e., indicator variables for membership in a random panel.

Formally, we prove the following lemma:
\begin{customlemma}{3}
    There is a polynomial-time sampling algorithm that, given a good pool $P$, produces a random panel $\Panel$ such that (1) $\mathbb{P}[i \in \Panel] = \pi_{i,P}$ for all $i \in P$, (2) $|\Panel| = k$, and (3) $\sum_{i : f(i)=v} \pi_{i,P} - |F| \leq |\{i \in \Panel \mid f(i)=v\}| \leq \sum_{i : f(i)=v} \pi_{i,P} + |F|$.
\end{customlemma}

To round the linear program, we use an iterative rounding procedure
based on the famous Beck-Fiala theorem~\cite{BF81}.
For ease of exposition, we first describe an algorithm for deterministic rounding and describe in the subsequent subsection how to turn it into a randomized rounding procedure.
From here on, we drop the index ``$P$'' from the marginal probabilities $\pi_{i,P}$, both for ease of notation and to emphasize that the lemma applies to any set of marginal probabilities adding up to $k$ (such other marginals might arise, say, from clipping and rescaling the $\pi_{i,P}$ if some of them are greater than 1).

\begin{lemma}
    \label{lem:beckfiala}
    For a pool $P$, let $(\pi_i)_{i \in P}$ be \emph{any} collection of variables in $[0,1]$ such that $\sum_{i \in P} \pi_i = k$.
    Then, we can efficiently compute a deterministic 0/1 rounding $(x_i)_{i \in P}$ such that $\sum_{i \in P} x_i = k$ and such that, for each feature-value pair $f,v$, \[\sum_{i \in P : f(i)=v} \pi_i - |F| \leq \sum_{i \in P : f(i)=v} x_i \leq \sum_{i \in P : f(i)=v} \pi_i + |F|.\]
\end{lemma}
\begin{proof}
We initialize $x_i \leftarrow \pi_{i, P}$, and the following inequalities are therefore satisfied:
\begin{align}
    \sum_{i \in P} x_i &= k \label{eq:bf1}\\
\sum_{i \in P : f(i)=v} x_i &= \sum_{i \in P : f(i)=v} \pi_{i,P} &&\forall f,v. \label{eq:bf2}
\end{align}
We then iteratively update the $x_i$ and maintain a set of equations that starts as the equations in \cref{eq:bf1,eq:bf2}, but from which we will iteratively drop some equations of type~\eqref{eq:bf2}.
Throughout this process, we maintain that the $x_i$ satisfy all remaining (i.e., not dropped) equations and that $x_i \in [0, 1]$ for all $i$.
We call $x_i \in (0, 1)$ \emph{active}; once an $x_i$ stops being active, it stays at its value $0$ or $1$ to the end of the rounding.
We continue our iterative process until no more active variables remain, at which point we return our 0/1 rounding.

Whenever the number of remaining equalities is lower than the number of active agents, the values $x_i$ for the active variables must be underdetermined by the equalities.
More precisely, after considering all inactive $x_i$ as constants, the space of remaining $x_i$ that satisfies the remaining equalities forms an affine subspace of non-zero dimension.
Since this subspace must intersect the boundary of the unit hypercube, there is a way of updating the $x_i$ such that all equalities are preserved, such that no inactive variable gets changed, and such that at least one additional variable becomes inactive (progress).\footnote{This step can be implemented in polynomial time by solving systems of linear equations.}

Else, we know that the number of active agents $n'$ is at most the number of remaining equalities $m$.
If $m = 1$, i.e., if \cref{eq:bf1} is the only remaining equation, there cannot be any active agents since \cref{eq:bf1} can only be satisfied if no $x_i$ or at least two $x_i$ are non-integer.
Thus, in the following, $m \geq 2$.
For any remaining equality of type~\eqref{eq:bf2} corresponding to some feature-value pair $f,v$, say that it \emph{ranges over} $t$ many active variables if there are $t$ many active variables $x_i$ such that $f(i)=v$.
Should any of the remaining constraints range over all $n'$ many active variables, then this constraint must be implied by constraint~\eqref{eq:bf1} and the values of the inactive variables.
We can thus drop the redundant constraint without consequences (progress), and repeat the iterative process.

If none of these steps apply, we show that some constraint of type~\eqref{eq:bf2} ranges over at most $|F|$ active variables:
Clearly, this is the case if $n' \leq |F|$, and furthermore if $n' = |F| + 1$ because we removed constraints of type~\eqref{eq:bf2} ranging over all active variables.
If $n' > |F| + 1$, note that every active agent appears in at most $|F|$ many equations of type~\eqref{eq:bf2}, at most one per feature.
It follows that the total number of active agents summed up over all remaining equalities of this type is at most $n' \, |F| < n' \, |F| - (|F| + 1) + n' = (n' - 1) \, (|F| + 1) \leq (m - 1) \, (|F| + 1)$, which implies that one of the $m-1$ equalities of type~\eqref{eq:bf2} ranges over less than $|F| + 1$ active variables.
Drop all such equalities (progress) and repeat.

Since $n' + m$ decreases in every iteration, this algorithm will produce a deterministic panel in polynomial time.
Since constraint~\eqref{eq:bf1} is never dropped, the panel size must be exactly $k$.
By how much might the equations of type~\eqref{eq:bf2} for a feature-value pair $f,v$ be violated in the result?
Clearly, they are maintained exactly up to the point where they are dropped.\footnote{We do not count if the equality was dropped because it was implied by constraint~\eqref{eq:bf1}, in which case it is preserved exactly throughout the rounding.}
From this point on, however, only $|F|$ many active variables could still change the value of $\sum_{i \in P: f(i)=v} x_i$.
Since each of these variables remains in its range $[0, 1]$ throughout the rounding process, the final $x_i$ must satisfy
\[\sum_{i \in P : f(i)=v} \pi_i - |F| \leq \sum_{i \in P : f(i)=v} x_i \leq \sum_{i \in P : f(i)=v} \pi_i + |F|.\qedhere \]
\end{proof}

\subsubsection{Randomizing the Beck-Fiala rounding}
\label{sec:beckfialarandom}
We give two methods of transforming the previous deterministic rounding algorithm into a randomized rounding algorithm.
To prove \cref{lem:beckfiala_random}, we can directly apply a result by Bansal~\cite{bansal2019} to our deterministic rounding procedure:
\begin{customlemma}{3}
    There is a polynomial-time sampling algorithm that, given a good pool $P$, produces a random panel $\Panel$ such that (1) $\mathbb{P}[i \in \Panel] = \pi_{i,P}$ for all $i \in P$, (2) $|\Panel| = k$, and (3) $\sum_{i : f(i)=v} \pi_{i,P} - |F| \leq |\{i \in \Panel \mid f(i)=v\}| \leq \sum_{i : f(i)=v} \pi_{i,P} + |F|$.
\end{customlemma}
\begin{proof}
We apply Theorem~1.2 by Bansal~\cite{bansal2019} to the deterministic rounding procedure of \cref{lem:beckfiala}.
To apply the theorem, we need to give a $\delta > 0$ such that, when there are $n'$ many active variables left, the number of remaining equalities in the next iteration is at most $(1 - \delta) \, n'$ constraints.
In \cref{lem:beckfiala}, we showed that $m$ is always set to a value of at most $n' - 1$.
Thus, for $\delta \coloneqq 1/n$, we get that $m \leq n' - 1 = (1 - 1/n') \, n' \leq (1 - 1/n) \, n'$ and can apply the theorem.
\end{proof}

While the previous algorithm runs in polynomial time, we found an alternative way of randomizing the rounding to be more efficient in practice.
This technique is based on na\"ive column generation, which is not guaranteed to run in polynomial time, but has the following advantages:
\begin{itemize}
    \item it uses linear programs rather than semi-definite programs,
    \item instead of a single random panel, the column generation (deterministically) generates a \emph{distribution} over panels, which allows us to analyze the distribution after a single run, and
    \item there is a continuous progress measure that allows us to stop the optimization process once we implement the $\pi_{i}$ with sufficient accuracy.
\end{itemize}
We describe this algorithm in the proof of the following version of \cref{lem:beckfiala_random}, which does not require polynomially-bounded runtime:
\begin{lemma}
    There is a sampling algorithm that, given a good pool $P$, produces a random panel $\Panel$ such that (1) $\mathbb{P}[i \in \Panel] = \pi_{i,P}$ for all $i \in P$, (2) $|\Panel| = k$, and (3) $\sum_{i : f(i)=v} \pi_{i,P} - |F| \leq |\{i \in \Panel \mid f(i)=v\}| \leq \sum_{i : f(i)=v} \pi_{i,P} + |F|$.
\end{lemma}
\begin{proof}%
First, note that we can strengthen \cref{lem:beckfiala} slightly by giving it an arbitrary vector $\vec{c} \in \mathbb{R}^{|P|}$ as part of its input and additionally requiring that $\langle\vec{c}, \vec{x}\rangle \geq \langle \vec{c}, \vec{\pi}\rangle$, where $\vec{x}$ is the vector of $x_i$ and $\vec{\pi}$ the vector of $\pi_i$.
This stronger statement follows from the same proof if we require every update of the $x_i$ to additionally maintain that $\langle\vec{c}, \vec{x}\rangle \geq \langle \vec{c}, \vec{\pi}\rangle$.
Since this intersects the non-zero dimensional affine subspace formed by the constraints with a half space that contains at least the current point $\vec{x}$, the resulting intersection is still unbounded, which means that we can find an intersection with the boundary of the hypercube.
We refer to this procedure as the ``modified \cref{lem:beckfiala}.''

Now, let $\mathfrak{B} \neq \emptyset$ be any set of panels satisfying the constraints of the lemma, possibly exponentially many. Consider the following linear program and its (simplified) dual:
\begin{align*}
    \text{PRIMAL($\mathfrak{B}$):} &&\text{DUAL($\mathfrak{B}$):}& \\
    \textit{minimize}~&\delta & \textit{maximize}~&\left(\sum_{i \in P} \pi_{i} \, z_i\right) - \hat{z} \\
    \textit{s.t.}~& \left| \pi_{i} - \sum_{B \in \mathfrak{B} : i \in B} \lambda_{B} \right| \leq \delta \quad \forall{i \in P} & \textit{s.t.}~&\sum_{i \in B} z_i \leq \hat{z} \quad \forall B \in \mathfrak{B} \\
    & \sum_{B \in \mathfrak{B}} \lambda_B = 1 && |z_i| \leq 1 \quad \forall i \in P \\
    & \delta \geq 0, \lambda_B \geq 0 \quad \forall B \in \mathfrak{B} &&
\end{align*}
The primal LP searches for a distribution over the panels $\mathfrak{B}$ such that the largest absolute deviation between the marginal $\sum_{B \in \mathfrak{B} : i \in B} \lambda_{B}$ and the target value $\pi_i$ of any $i \in P$ is as small as possible.
Let $\overline{\mathfrak{B}}$ denote the set of panels that can be returned by the modified \cref{lem:beckfiala}, for any vector $\vec{c}$ in its input.

\textbf{Observation 1: For any $\mathfrak{B} \neq \emptyset$, the LP has an objective value $\mathit{obj}(\mathfrak{B}) \geq 0$.} Indeed, in the primal, the objective value is clearly bounded below by $0$, and the LP is feasible for any distribution over $\mathfrak{B}$ and large enough $\delta$.
By strong duality, the dual LP must have the same objective value.

\textbf{Observation 2: $\mathit{obj}(\overline{\mathfrak{B}}) = 0$.} For the sake of contradiction, suppose that the objective value was strictly positive, i.e., that $\vec{\pi}$ does not lie in the convex hull of $\overline{\mathfrak{B}}$.
Then, there must be a plane separating $\vec{\pi}$ from this convex hull, and an orthogonal vector $\vec{c}$ such that $\langle \vec{c}, \vec{\pi}\rangle > \langle \vec{c}, \vec{x} \rangle$ for any $\vec{x}$ corresponding to a panel in $\overline{\mathfrak{B}}$. 
Applying the modified \cref{lem:beckfiala} with this vector $\vec{c}$ would lead to a contradiction.

Consider \cref{alg:col}, which iteratively generates a subset $\mathfrak{B} \subseteq \overline{B}$ by column generation.
\begin{algorithm}[h]
\DontPrintSemicolon

$\mathfrak{B} \leftarrow \{\text{result of running modified \cref{lem:beckfiala} with arbitrary $\vec{c}$}\}$\;
\While{$\mathit{obj}(\mathfrak{B}) > 0$}{\label{lin:while}
fix optimal values $z_i, \hat{z}$ for DUAL($\mathfrak{B}$)\;
$B \leftarrow \text{result of running modified \cref{lem:beckfiala} with $\vec{c}$ as the vector of $z_i$}$\;\label{lin:sep}
$\mathfrak{B} \leftarrow \mathfrak{B} \cup \{B\}$
}
\Return $\mathfrak{B}$
\caption{Column generation}\label{alg:col}
\end{algorithm}

\textbf{Observation 3: \cref{alg:col} terminates.}
It suffices to show that, in \cref{lin:sep}, the generated panel $B$ is not yet contained in $\mathfrak{B}$ since, then, the size of $\mathfrak{B}$ grows in every iteration and is always upper-bounded by the finite cardinality of $\overline{\mathfrak{B}}$.
By the definition of the modified \cref{lem:beckfiala}, $B$ always satisfies $\sum_{i \in B} z_i \geq \sum_{i \in P} \pi_{i} \, z_i$.
However, since the objective value is positive, any $B' \in \mathfrak{B}$ satisfies $\sum_{i \in P} \pi_i \, z_i > \hat{z} \geq \sum_{i \in B'}z_i$, which shows that $B \notin \mathfrak{B}$.

Once \cref{alg:col} terminates with a set $\mathfrak{B}$, we know that $\mathit{obj}(\mathfrak{B}) = 0$, which means that, by solving PRIMAL($\mathfrak{B}$), we obtain a distribution over valid panels that implements the marginals $\pi_i$, which concludes the proof.
\end{proof}

In practice, it makes sense to exit the while loop in \cref{lin:while} already when $\mathit{obj}(\mathfrak{B})$ is smaller than some small positive constant, which guarantees a close approximation to the marginal probabilities while reducing running time and preventing issues due to rounding errors.

\subsection{Proof of \cref{thm:main}} \label{sec:app_thmmain}
\begin{customthm}{1}
Suppose that $\alpha \to \infty$ and $n_{f,v} \geq n/k$ for all feature-value pairs $f,v$.
Consider a sampling algorithm that, on a good pool, selects a random panel, $\Panel$, via the randomized version of \cref{lem:beckfiala_random}, and else does not return a panel. 
This process satisfies, for all $i$ in the population, that
    \[ \mathbb{P}[i \in \Panel] \geq (1 - o(1)) \, k/n. \]
All panels produced by this process  satisfy the quotas $\ell_{f,v} \coloneqq (1 - \alpha^{-.49}) \, k \, n_{f,v} / n - |F|$ and $u_{f,v} \coloneqq (1 + \alpha^{-.49}) \, k \, n_{f,v} / n + |F|$ for all feature-value pairs $f,v$.
\end{customthm}
\begin{proof}
    The claim about the quotas immediately follows from \cref{lem:beckfiala_random} and the definition of a good pool. Concerning the selection probabilities,
    \begin{align*}
        &\mathbb{P}[i \in \Panel] = \sum_{\mathclap{\substack{\text{good pools $P$}\\i \in P}}} \mathbb{P}[i \in \Panel \mid \Pool = P] \, \mathbb{P}[\Pool = P] 
        = \!\!\!\!\!\!\!\!\sum_{\substack{\text{good pools $P$}\\i \in P}} \frac{k \, a_i}{\sum_{j \in P} a_j} \, \mathbb{P}[\Pool = P]. \\
        \intertext{Since $\sum_{j \in P} a_j \leq r / (1 - \alpha^{-.49})$ for good pools, we continue}
        \geq{} &(1 - \alpha^{-.49}) \, k / (r \, q_i) \sum_{\mathclap{\substack{\text{good pools $P$}\\i \in P}}} \mathbb{P}[\Pool = P] = (1 - \alpha^{-.49}) \, \frac{k}{r \, q_i} \, \mathbb{P}[i \in \Pool \land \text{$\Pool$ is good}] \\
        ={} &(1 - \alpha^{-.49}) \, \frac{k}{r \, q_i} \, \underbrace{\mathbb{P}[\text{$\Pool$ is good} \mid i \in \Pool]}_{\text{$\in 1 - o(1)$ by \cref{lem:goodpools}} } \, \underbrace{\mathbb{P}[i \in \Pool]}_{= q_i \, r/n}
        \in (1 - o(1)) \, \frac{k}{n}. \qedhere
    \end{align*}
\end{proof}

\section{Supplementary Material for \cref{sec:learning}} \label{sec:app_learning}
\paragraph{Participation Model} Let $y_i = 1$ for agents who would join the pool if invited, and $y_i = 0$ for agents who would not. We want to predict $q_i = \mathbb{P}[y_i = 1]$ for all agents in the pool. To do so, we learn the following parametric model, which describes the relationship between an agent's feature vector $F(i)$ and value of $q_i$.
\[q_i = \beta_0 \prod_{f \in F} \beta_{f,f(i)}
\]

This type of generative model describes a decision process known as \emph{simple independent action}~\cite[as cited in \cite{Weinberg86}]{Finney71}.
To express this model in a more standard form, let $x_i$ be a vector describing agent $i$'s values for all features in $F$, where each index $j$ of $x_i$ corresponds to a feature-value $f,v$ and contains a binary indicator of whether agent $i$ has value $v$ for feature $f$. Let $M$ be the length of $x_i$, where $M =  1+ \# feature$-$values$. We then reshape parameters $\beta_0$, $\beta_{f,v}$ for all $f,v$ into a parameter vector $\boldsymbol{\beta}$ of length $M$, and correspondingly, $x_i$ must have value 1 at its first index for all agents $i$, corresponding to the parameter $\beta_0$. We can then write an equivalent version of our model in more standard form. Note that $q_i$ is technically a function of $x_i,\boldsymbol{\beta}$, but we omit this notation for simplicity.
\[q_i = \prod_{j \in [M]} \boldsymbol{\beta}_j^{x_{i,j}}\]

\paragraph{Maximum Likelihood Estimation with Contaminated Controls}
To estimate the parameters $\boldsymbol{\beta}$ of this model on fixed pool $P$ and fixed background sample $B$, we apply the estimation methods in Section 3 of Lancaster and Imbens \cite{lancaster1996case}. We use the objective function in Equation 3.3, which is designed to perform maximum-likelihood estimation (MLE) in the setting of contaminated controls. Let $z_i$ be an indicator such that $z_i = 1$ for $i \in P$ and $z_i = 0$ for $i \in B$. Let $w_i$ be the weight of agent $i \in B$ (for details on these weights, see \Cref{app:empirical}). Recall that $\overline{q}$ is the average participation probability in the underlying population. Then, the likelihood function $L(\boldsymbol{\beta})$ that we would maximize to directly learn our model is
\[L(\boldsymbol{\beta}) = \sum_{i \in B \cup P} \left(z_i \sum_{j \in [M]}\left(x_{i,j} \log \boldsymbol{\beta}_j\right) - w_i \log\left(\overline{q} \, |B|/|P| + \prod_{j \in [M]} \boldsymbol{\beta}_j^{x_{i,j}}\right)\right)\]
Unfortunately, $L(\boldsymbol{\beta})$ is not obviously concave in $\boldsymbol{\beta}$. To get around this, we re-parameterize our model such that we can instead learn the \textit{logarithms} of our parameters. Defining a new parameter vector $\theta$ such that $\theta_j = \log(\boldsymbol{\beta}_j)$ for all $j \in [M]$, we can rewrite our model equivalently as the exponential model.
\[q_i = \prod_{j \in [M]} \boldsymbol{\beta}_j^{x_{i,j}} = \exp{\left(\log \left(\prod_{j \in [M]} \boldsymbol{\beta}_j^{x_{i,j}}\right)\right)}  = \exp{\left(\sum_{j \in [M]} x_{i,j} \log(\boldsymbol{\beta_j}) \right)} = e^{\theta x_i}\]
By Equation 3.3 in Lancaster and Imbens \cite{lancaster1996case}, the likelihood function $L'(\theta)$ we maximize is now the following. By \Cref{thm:app_concave}, this objective function is concave, so it can therefore be maximized efficiently (under the constraint that $\theta \leq 0$).
\begin{equation}L'(\theta) = \sum_{i} \left(z_i \theta x_i - w_i \log\left(\overline{q} \, |B|/|P| + e^{\theta x_i}\right)\right)\label{eq:mle}\end{equation}
\begin{theorem}\label{thm:app_concave} The log-likelihood function for the simple independent action model under contaminated controls is concave in the model parameters.
\end{theorem} 
\begin{proof}
The first term of the sum is linear, so both concave and convex. The second term is concave by \Cref{lemma:concave_mle}, 
\end{proof}

\begin{lemma} \label{lemma:concave_mle}
Let function $f(\theta) = -\log(c + e^{\theta X})$, where $c > 0$ is a constant. $f$ is concave. 
\end{lemma}
\begin{proof}
The $i,j$th term of the Hessian matrix  $H$ of $f$ can be written as
    \[H_{i,j} = -X_iX_j  \frac{c e^{\theta X}}{\left(c + e^{\theta X}\right)^2}\]
Now, let $\psi = \frac{\sqrt{c e^{\theta X}}}{c + e^{\theta X}}$.
Noting that $X$ is considered a column vector, we can then rewrite the Hessian in terms of $\psi$ as $H = -(\psi X)(\psi X)^T$. In words, the negative Hessian can be written as the outer product of the vector $\psi X$ with itself. Therefore, the negative Hessian is positive semi-definite, and the Hessian is negative semi-definite, implying that $f$ is concave.
\end{proof}

\paragraph{Discussion of Methods}
The reader may note that we treat $\overline{q}$ as a known constant in our estimation, but the objective function we use from Lancaster and Imbens is designed for the setting in which $\overline{q}$ is a variable. There is precedent in the literature for doing so \cite{ward2009presence}. As Lancaster and Imbens discuss, using $\overline{q}$ as a constant rather than a variable when maximizing Equation 3.3 introduces issues of over-parameterization, because it is not enforced that the average $q_i$ over the population be $\overline{q}$.
While we cannot estimate $q_i$ values for the entire population for lack of data, it would be a worrying sign if the average $q_i$ over the \emph{background sample}, a uniform sample from the population, was far from our assumed $\overline{q}$.
However, we find that the average of our estimated $q_i$ values over the background sample is 2.9\%, which matches $\overline{q} = 2.9\%$.

\section{Supplementary Material for \cref{sec:empirical}}
\label{app:empirical}

For estimation, we use two datasets. For our positively-labeled data, we use the set of pool members from the UK Climate Assembly (for details, see \Cref{sec:app_ukassembly}). For our background sample, we use the European Social Survey (ESS), which serves as an unlabeled uniform sample of the population. 
\subsection{Climate Assembly UK Details \& Pool Dataset} \label{sec:app_ukassembly} 
Our pool dataset contains the agents from the pool of the \textit{Climate Assembly UK}, a national-level sortition panel on climate change held in the UK in 2020. We use ``panel'' to refer to the group of people who deliberate, and ``assembly'' to refer to the actual deliberation step. The panel for this assembly as selected by the \textit{Sortition Foundation}, a UK-based nonprofit that selects sortition panels. A document by the Sortition Foundation gives the following description of this assembly:\footnote{\url{https://docs.google.com/spreadsheets/d/1kwgOpxMX4pwR3Myu4pXku4gjcnOS53bPOKwOGjZNxyI/edit\#gid=0}}
\begin{quote}
    \textit{This Citizens' Assembly will meet across four weekends in early 2020 to consider how the UK can meet the Government’s legally binding target to reduce greenhouse gas emissions to net zero by 2050. The outcomes will be presented to six select committees of the UK parliament, who will form detailed plans on how to implement the assembly's recommendations. These plans will be debated in the House of Commons.}
\end{quote}
In the formation of the panel for this assembly, 30\,000 letters were sent out inviting people to participate. Of these letter recipients, 1\,727 people entered the pool, and 110 people were selected for the panel. The features and corresponding sets of values used for this panel are described in \Cref{tab:uk_minipublic_fv}.

\begin{table}[h!]
        \centering
        \begin{tabular}{lp{9.5cm}}
            \toprule
            Feature ($f \in F$) & Values ($V_f$) \\
            \midrule
             Gender & Male, Female, \sout{Other} \\
             Age & 16-29, 30-44, 45-59, 60+\\
             Region & North East, North West, Yorkshire and the Humber, East Midlands, West Midlands, East of England, London, South East, South West, Wales, Scotland, Northern Ireland\\
             Education Level & No Qualifications/Level 1, Level 2/Level 3/Apprenticeship/Other, Level 4 and above\\
             \sout{Climate Concern Level} & \sout{Very concerned, Fairly concerned, Not very concerned, Not at all concerned, Other} \\
             Ethnicity & White, Black or ethnic minority (BAME)\\
             Urban / Rural & Urban, Rural \\
             \bottomrule\\
        \end{tabular}
        
        \caption{Climate Assembly UK features and values.}
        \label{tab:uk_minipublic_fv}
    \end{table}
Those with value \textit{Other} for gender were dropped from the pool data because an equivalent value could not be constructed in the ESS data. This resulted in us dropping 12 people out of the original 1727, for a pool dataset of final size 1715. Note that dropping these people did not affect our estimate of $\overline{q}$ --- before and after dropping these agents, it was 2.9\%. The Climate Concern Level feature was dropped altogether from the set of features used for analysis because there were too few people in the pool with value \textit{Not at all concerned} to give these agents proportional representation on the panel. 

Due to privacy agreements between the Sortition Foundation and the pool members, we are unable to share this dataset.

\subsection{Background Data} \label{sec:app_data}
We define the size of the ESS dataset to be the sum of the weights of the agents within it.\footnote{This sum should ideally be equal to the number of people in the ESS data, but because we drop a few people, the sum of weights no longer exactly equals the number of people.} For details on weights, see the \textit{Re-weighting} paragraph of this section. In order to use this data as our background sample, we construct feature vectors for each person in the ESS data that correspond to those used in Climate Assembly UK, as defined in \Cref{tab:uk_minipublic_fv}. 

In this section, we describe how we constructed the variables corresponding to the features and their values as specified by the Sortition Foundation. We dropped 44 people out of the original 1959 people in the ESS dataset, and we briefly discuss this decision and its implications. Finally, we describe how we re-weighted the ESS data to correct for sampling and non-response bias to approximate the scenario in which the surveyed individuals were uniformly sampled from the population. This step is important because, in our $q_i$ estimation procedure, we assume that our background sample is uniformly sampled.

\paragraph{Variable construction} Fortunately, the ESS data contained variables and categories that either exactly or very closely corresponded to the features and values specified by the Sortition Foundation. Essentially the only modification to the ESS data we made to construct valid feature vectors was the aggregation over categories in the \textit{Education Level} and \textit{Urban/Rural} ESS variables, which were broken down into more fine-grained categories than those specified in \Cref{tab:uk_minipublic_fv}. In general, for features with values containing the value ``other'', missing data was assigned the value ``other''. Below is a table showing which variables and values from the ESS data were used to construct each feature from the Climate Assembly UK. Exact details on how these variables were used is documented in the code (see \Cref{sec:implementationdetails} for reference to readme).
\begin{center}
\begin{tabular}{l|l}
    Feature (Climate Assembly UK) & Variable (ESS raw data)  \\
    \hline
    Gender & gndr \\
    Age & agea\\
    Region & region \\
    Education Level & edulvlb\\
    Climate Concern Level & wrclmch\\
    Ethnicity & blgetmg\\
    Urban/Rural & domicil\\
\end{tabular}
\end{center}

\paragraph{Dropping people} As described in \Cref{tab:uk_minipublic_fv}, the Climate Assembly's youngest valid age category was 16-29. We therefore dropped all four people in the ESS data who were under 16 years old. Dropping people who fall outside our demographic ranges of interest is not a problem for weights, because the weights of all people of interest (who we want to be fair to) will remain the same relative to each other, and we care only about the composition of this relevant population.\\
There were an additional 40 people who may have been within our demographic range of interest, but who were missing age, race, or urban/rural data. Among these 40 people, 33, 6, and 4 people did not have data for variables corresponding to the features \textit{age}, \textit{ethnicity}, and \textit{urban / rural}, respectively. While dropping these people could affect the weighting scheme, the distribution of weights of those dropped is strongly right-skewed, meaning that those who we dropped belong to groups that tended to be oversampled in the ESS data. These people are therefore likely more numerous in the ESS data overall, and dropping some of them will have a smaller proportional effect.\\
Finally, the ESS did not permit people to answer ``other'' for gender, a category permitted on the Sortition Panel. Without any way to construct the \textit{gender = other} feature-value in the ESS data, we dropped the members of the Climate Assembly pool with this feature-value.

\paragraph{Re-weighting} The ESS recommends re-weighting their data to correct for bias, and they provide multiple sets of possible weighting schemes for doing so\footnote{\url{https://www.europeansocialsurvey.org/docs/methodology/ESS_weighting_data_1.pdf}}. Of the provided options, we elected to apply the Post-Stratification Weights, because these weights account for not only sampling bias, but also non-response bias, by incorporating auxiliary information from other demographic surveys. By this weighting scheme, each person in the ESS data is given a weight $w_i$, representing how much that person should count in the analysis of the ESS data, where the weights are normalized to 1. This weight is encoded in the ESS data as `pspwght'.

\paragraph{Estimation of $\boldsymbol{\overline{q}}$} We bolster the identification of our model with an estimate of $\overline{q}$, the rate of true positives in the population. In our setting, this is the number of people who would ultimately enter the pool if invited. We estimate $\overline{q}$ in Climate Assembly UK data roughly as the fraction of people who joined the pool (1\,715) out of those who were invited (30\,000). These numbers seem to imply that the $\overline{q} \approx 1\,715/30\,000 = 5.7\%$. However, there is a complication: each letter is sent to a \textit{household}, rather than an individual, and any eligible member of an invited household may join the pool. Using the ESS data, we compute (see below) the average number of eligible panel participants per household to be 2.00, implying that in reality, 60\,000 eligible people were invited to participate in the pool. As a result, we estimate $\overline{q}$ to be $\overline{q} = 1\,715/60\,000 \approx 2.9\%$.

Let $ESS$ be the set of agents in the cleaned ESS data. Computing the average number of eligible panel participants per household from the ESS data is not entirely trivial, because sampling \textit{people} uniformly (or in the case of the ESS, approximating uniform sampling by re-weighting) is biased toward larger households. To account for this, for each person $i \in ESS$, we scale their weight $w_i$ by the inverse of the number of eligible people in their household, $\mathit{householdsize}_i$. Then,
\begin{align*}
    \text{average number of eligible people per household} &= \frac{\sum_{i \in ESS} \left(\frac{w_i}{\mathit{householdsize}_i}\right) \cdot \mathit{householdsize}_i}{\sum_{i \in ESS} \left(\frac{w_i}{\mathit{householdsize}_i}\right)}
\end{align*}
We compute $householdsize_i$ for each person $i \in ESS$ using the weighted ESS data. Age is the only feature from the UK Climate Assembly for which the ESS data may contain values rendering a person ineligible (specifically, the ESS data surveys people down to age 15, while the climate assembly accepted only those over 16). To count the number of people in each household who are eligible, we use variables `agea', `pspwght', and `yrbrn2-12', which describe the ages of person $i$'s household members (up to 12 household members).

\subsection{Implementation Details}
\label{sec:implementationdetails}
Our experiments were implemented in Python, using PyTorch for the MLE estimation and Gurobi for solving the linear programs in the column generation.
Our code is contained in the supplementary material and will be made available as open source when published.
The file ``README.md'' in the code gives detailed instructions for reproducibility.

We found the log-likelihood presented in \cref{eq:mle} to be easy to maximize.
For accuracy, we chose a small step size of $10^{-5}$ and a large number $10^5$ of optimization steps.
The final objective was $4157.32345$, and objective changes between iterations $20\,000$ and $100\,000$ were less than $3 \times 10^{-6}$.

Our experiments were run on a 13-inch MacBook Pro (2017) with a 3.1 GHz Dual-Core i5 processor.
Optimizing the log-likelihood took 46 seconds.
Running the column generation took 38 minutes to reach the desired accuracy of $10^{-6}$, which is much smaller than the smallest $\pi_{i,P}$ at around $2\%$.
For the version including climate concern, MLE estimation took 37 seconds reaching a log-likelihood of $4601.01427$, and column generation took 26 minutes.

Sampling 100\,000 pools each and simulating our algorithm for the end-to-end experiments took 30 minutes for $r=10\,000$, 55 minutes for $r=11\,000$, 61 minutes for $r = 12\,000$, 76 minutes for $r=15\,000$, and 95 minutes for $r=60\,000$.
All running times should be seen as upper bounds since other processes were running simultaneously.
Sampling the same number of pools for the case including the climate concern feature took around 410 minutes for $r=600\,000$.
The equivalent experiments with the greedy algorithm took around 19 hours (floor and ceiling quotas) and around 12 hours (no quotas).

\subsection{Results and Validation of $\boldsymbol{\beta},\boldsymbol{q_i}$ Estimation}
\label{sec:app_betavalidation}
\paragraph{Pool and Background Data Composition} First, we examine the frequency at which each feature-value occurs in the pool and the background data. As shown in the figure below, those with the most education are highly over-represented in the Climate Assembly UK pool compared to the background sample, and people with low education are under-represented. Similarly, we see men are slightly over-represented in the pool, and increasing age also seems to increase likelihood of entering the pool.
\begin{center}
    \includegraphics[width=\textwidth]{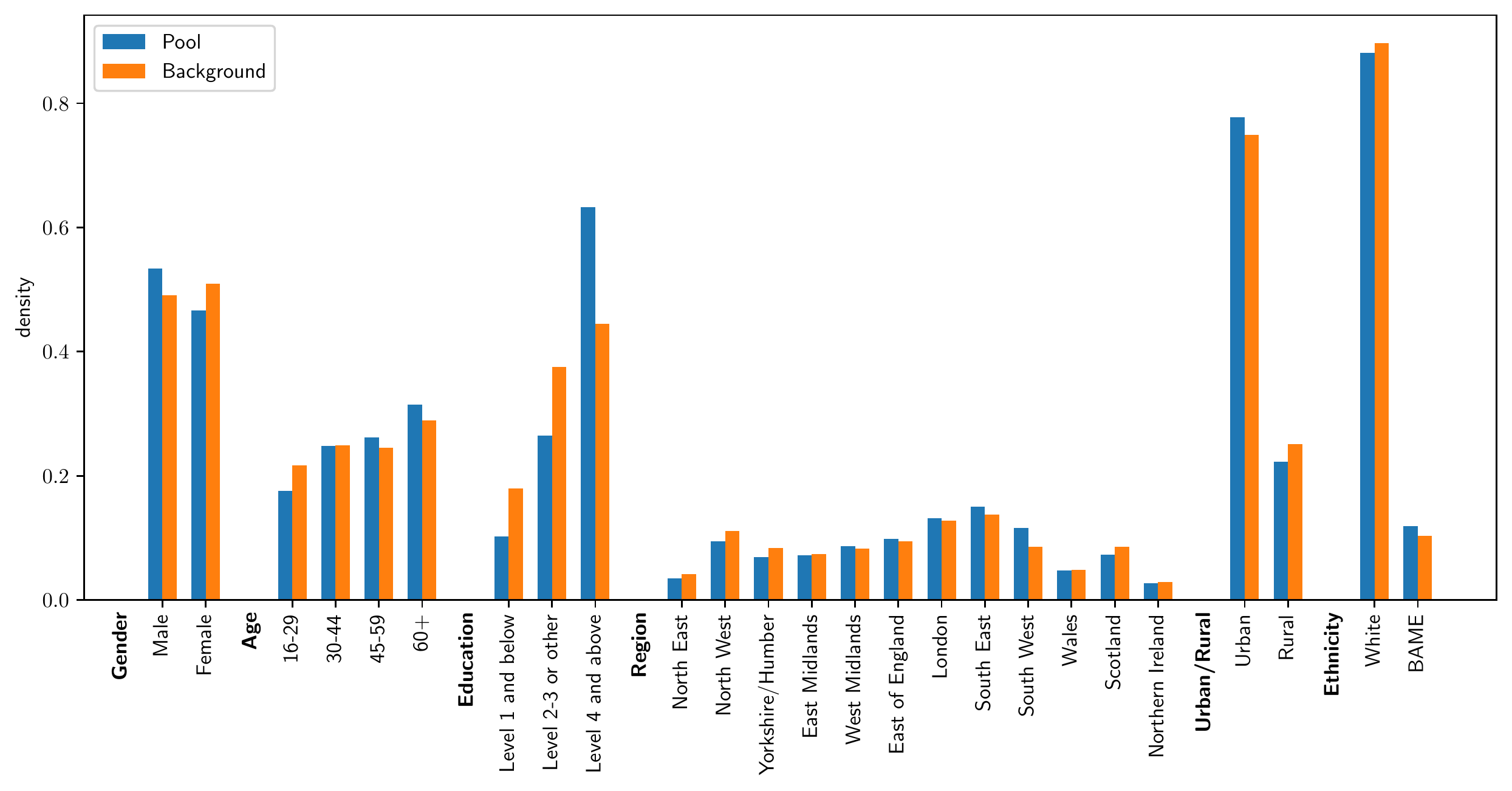}
\end{center}

\paragraph{Estimates of $\boldsymbol{\beta}$}
We find that $\beta_0 = 8.8\%$, meaning that all agents participate with a baseline probability of 8.8\%. In the figure below are estimates of $\beta_{f,v}$ for all feature-values $f,v$. Recall that $1-\beta_{f,v}$ can be interpreted as the probability of not participating due to having value $v$ for feature $f$; in other words if $\beta_{f,v}$ is 1, then feature-value $f,v$ has no adverse effect on whether a person participates.

Notably, these $\beta$ estimates are consistent with the composition of the pool compared to the background data. For example, people of increasing age were increasingly over-represented in the pool compared to the background data, and we see here that $\beta$ associated with age increase with increasing age. Similarly, we see that having low education greatly diminishes a person's likelihood of participation, corresponding to the observation that the pool contained a disproportionately low number of people with the two lower levels of education. In fact, one can confirm that across all feature-values, $\beta$ values correspond with the composition of the pool data compared to the background data, indicating that the $\beta$ values learned with our model are a good fit to the data used to learn them.
\begin{center}
    \includegraphics[width=\textwidth]{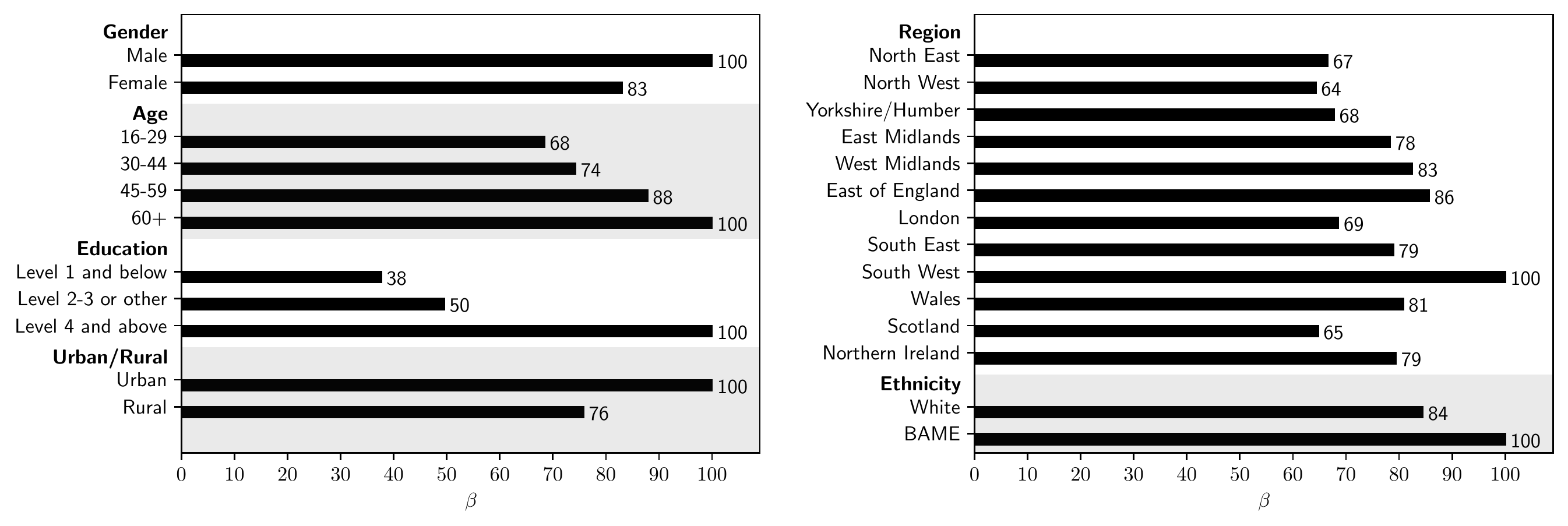}
\end{center}

\paragraph{Estimates of $\boldsymbol{q_i}$}
We compute our $q_i$ estimates based on $\beta$ estimates according to the model in \Cref{sec:app_learning}. We get the following distributions of $q_i$ values in the pool and background datasets.

The data shown in this plot is limited to density of $q_i$ values between 1\% and 8\%, because bins outside this range contain fewer than 7 people, and are withheld to avoid potential privacy issues. Less than $0.3\%$ of agents in either dataset are excluded for this reason.

\begin{center}
    \includegraphics[width=0.6\textwidth]{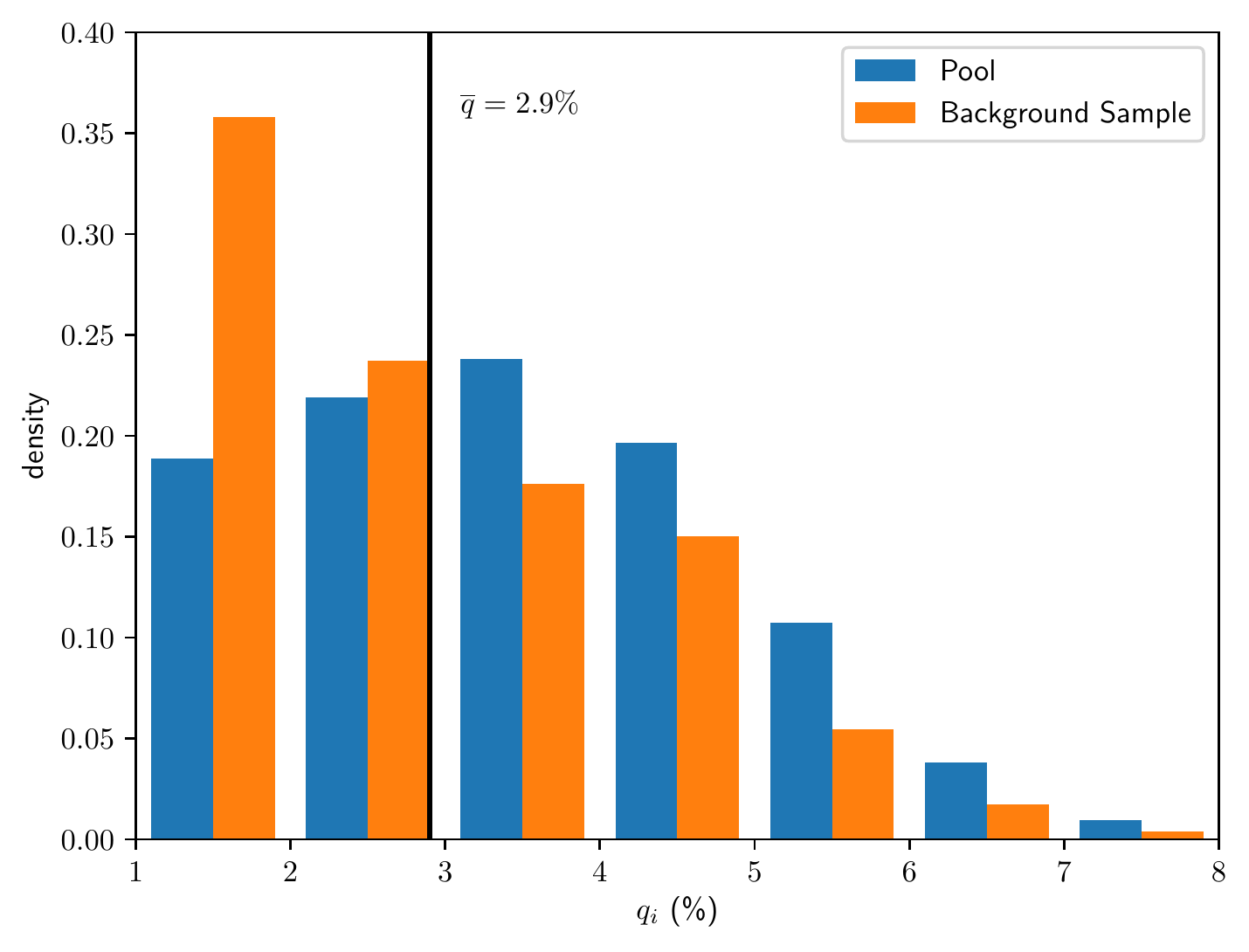}
\end{center}

Not very surprisingly, we find that the pool overrepresents agents with higher participation probability with respect to their share in the background sample.

\paragraph{Test for Calibration of $\boldsymbol{q_i}$ Estimates}
To validate whether our model fits the data well, we form a \emph{hypothetical pool} by imagining that the weighted background sample was selected as the set of recipients and that the members of this set participate with our estimated probability $q_i$.
For some attributes that agents might have or not have, the expected number of agents in the hypothetical pool with this attribute is
\[ \sum_{i \in B : \text{$i$ has attribute}} q_i.\footnote{Of course, all operations on the background sample respect the weights, which we ignore here for the sake of clarity.}\]
Since the set of invitation recipients to the Climate Assembly and the background sample are both assumed to be representative samples of the population, we would expect the above sum to be (close to) proportional to the fraction of pool members with this attribute\,---\,at least if the model fits the data well.

For instance, this idea allows us to re-examine the previous plot of $q_i$ values by letting the orange bars not denote the (scaled) \emph{number} of members in the background sample with $q_i$ in the right range, but instead the (scaled) \emph{sum of $q_i$ values} of members in the background sample with $q_i$ in this range.

The fact that these distributions align fairly well can be seen as our $q_i$ passing a sort of calibration test\,---\,of those agents with a certain $q_i$ value, roughly a $q_i$ proportion would participate when invited.
Relative to our background sample, the Climate Assembly pool does not seem to untypically skew towards agents with low or high values of $q_i$.

Once again, for privacy reasons we display frequencies of $q_i$ values only between 1\% and 8\%. Once again, less than $0.3\%$ of agents in either dataset are excluded for this reason.
\begin{center}
    \includegraphics[width=0.6\textwidth]{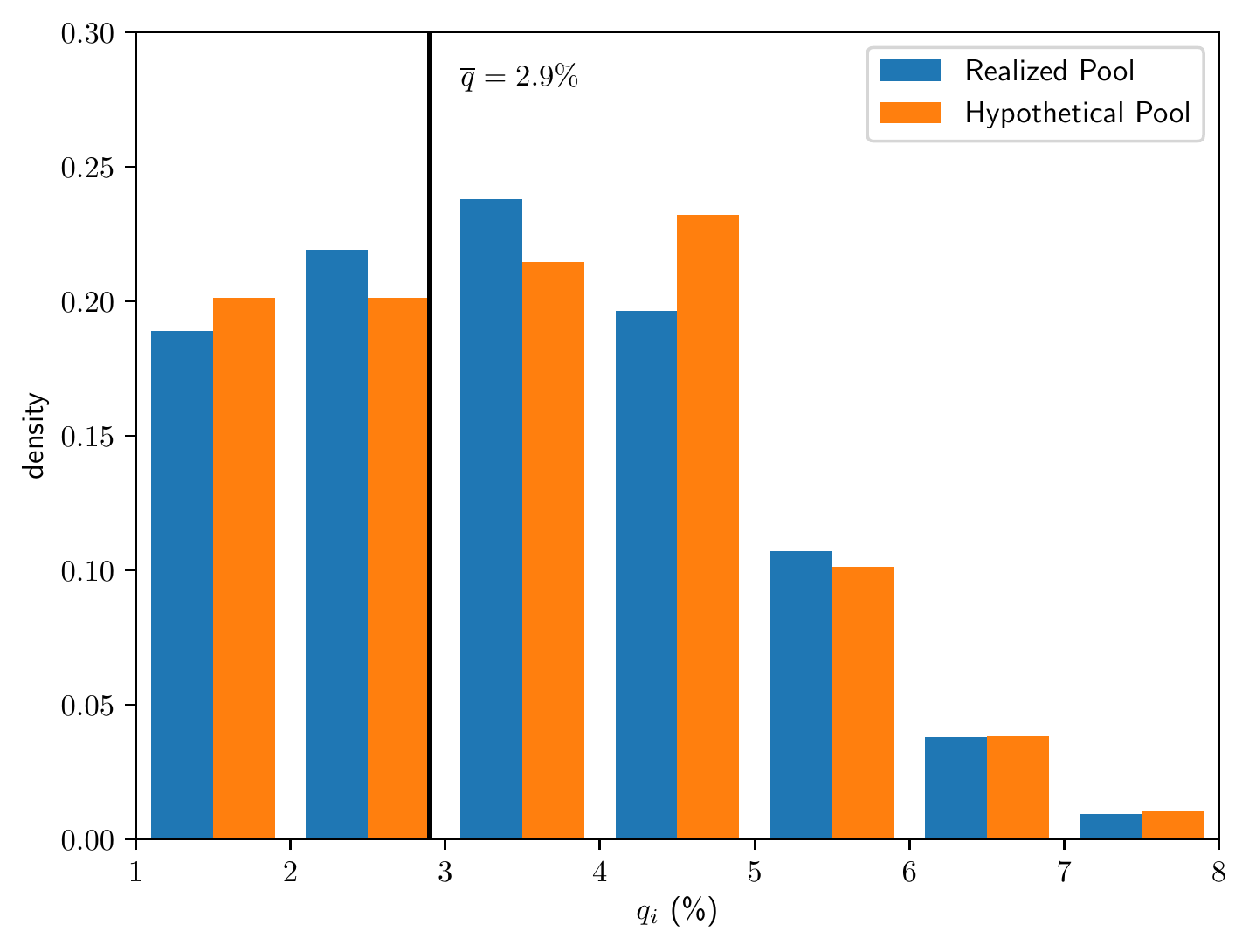}
\end{center}

\paragraph{Comparison of Realized Pool Composition and Hypothetical Pool Composition}
We now plot the same comparison between the Climate Assembly pool and the hypothetical pool but for the prevalence of each feature-value pair.

The figure below shows that if our $\beta$ estimates and the $q_i$ estimates they yield are true for members of the population, then if we sampled the underlying population as was done to form the Climate Assembly UK pool, we would get in expectation a pool that looks almost identical to realized pool. This illustrates in another way that our $\beta$ estimates are a good fit to the data we provided.
\begin{center}
    \includegraphics[width=\textwidth]{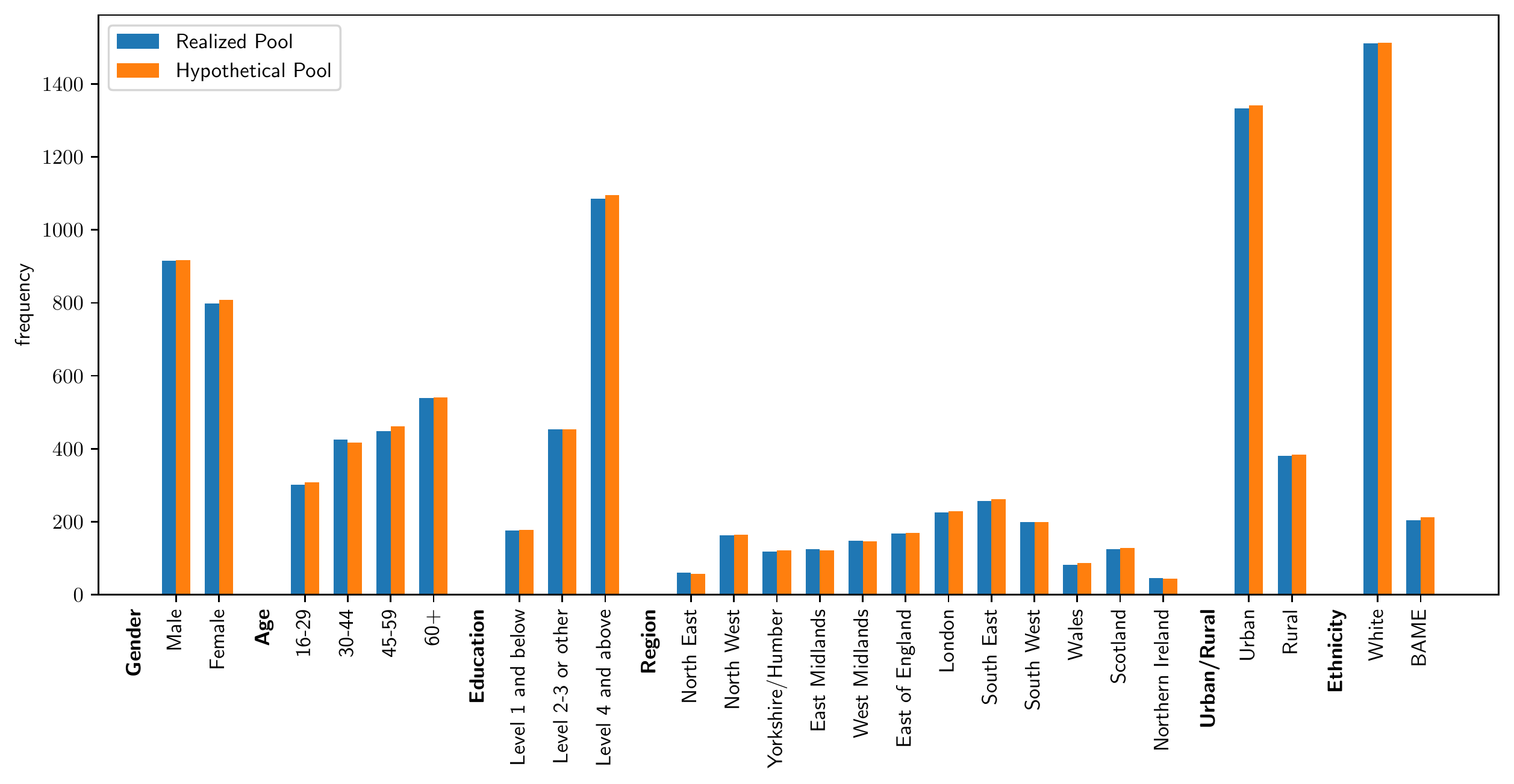}
\end{center}

\paragraph{Testing Model Capture of 2-Correlations}
Our model assumes that each feature-value affects people's probability of participating independently of all other feature-values. This analysis tests whether this causes our model to severely misjudge the participation probability for some group defined by the intersection of \emph{two} feature-value pairs, again comparing the prevalence of these groups in the Climate Assembly UK pool vs. the hypothetical pool that would be drawn from a population with the same composition as the background sample. On the plot below, each point represents an intersection of two feature-values. Each point's $x$ and $y$ coordinates are the fraction of people with that intersection in the Climate Assembly UK pool and the fraction of the hypothetical pool, respectively. We would hope for this relationship to be exactly linear, illustrating that each pair of feature-values occurs at the same rate in the real vs. hypothetical pool.
\begin{center}
    \includegraphics[width=.75\textwidth]{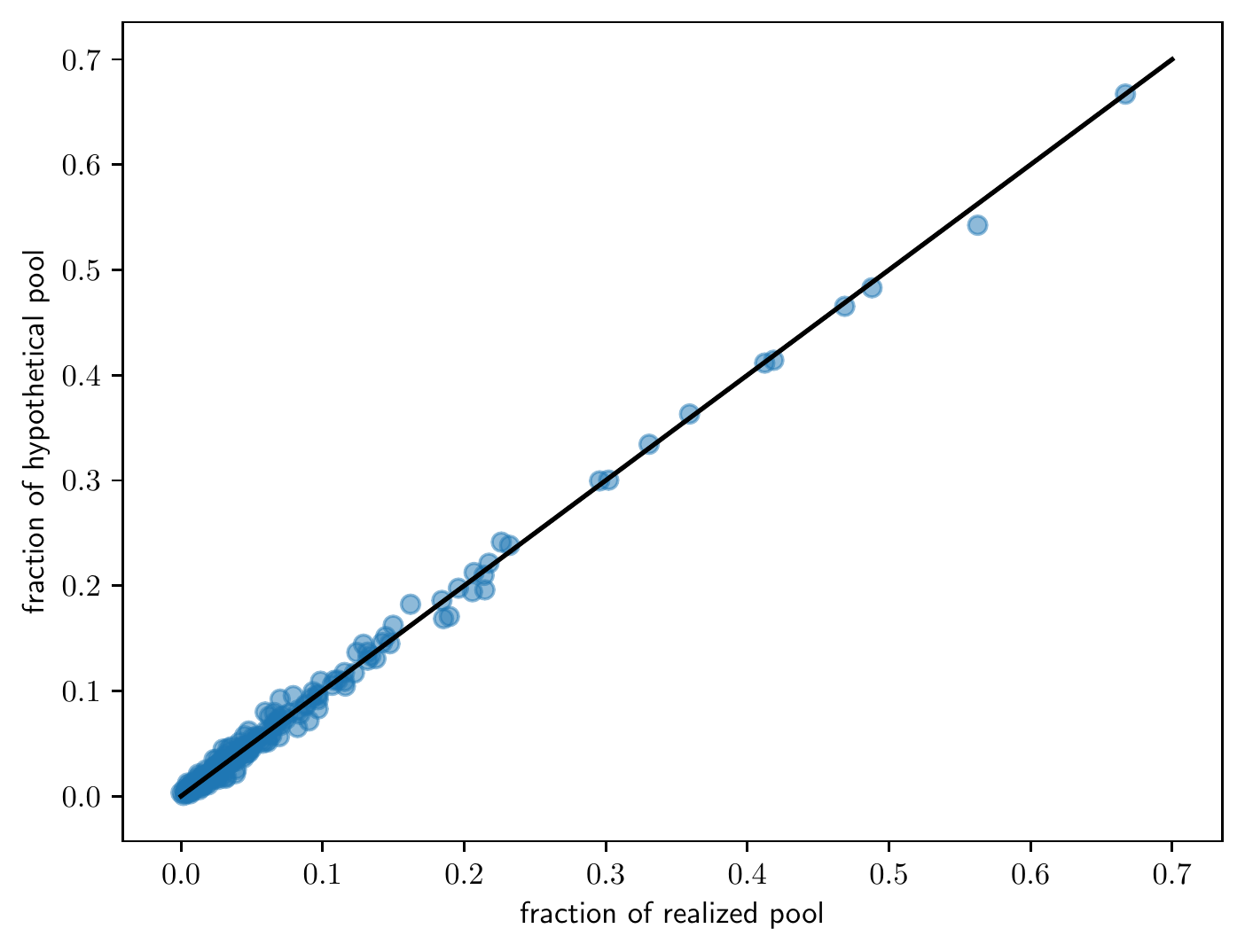}
\end{center}

\subsection{Details on End-To-End Experiment}
\label{app:endtoenddetails}
As described in the body of the paper, we generate a synthetic population by scaling up the ESS participants to a population of 60 million individuals.
The number of copies of a participant is proportional to their weight in the ESS, and is rounded to an integer using the Hamilton apportionment method.
100\,000 times per experiment, we select a set of letter recipients of size $r$ uniformly from the population, and flip a biased coin with probability $q_i$ for each letter recipient to determine whether she joins the pool.
For each pool, we then obtain the selection probabilities of the pool members conditioned on this being the pool (or an unbiased estimate of these probabilities):
\begin{itemize}
    \item For our algorithm, we check whether the pool $P$ is good. If the pool is not good, we (conservatively) assume that no panel is returned, and that pool members have zero probability of being selected. Else, we return the selection probabilities $\pi_{i,P}$.
    \item We use the implementation of the greedy algorithm developed by the Sortition Foundation and available at \url{https://github.com/sortitionfoundation/stratification-app/tree/4a957359b708a327aad0103ab2a59d061aeaeeb4}. Since we do not have a closed form for individual selection probabilities, we run the greedy algorithm 10 times and report the average time that each pool member was selected. While these estimates of selection probabilities are noisy, they are unbiased estimates of the end-to-end probability and independent between pools. Thus, the noise largely averages out over the 100\,000 random pools. In no case did the greedy algorithm fail to satisfy the quotas.
\end{itemize}
Each point in the diagrams corresponds to one agent in the ESS sample and indicates this agents' $q_i$ as well as the average selection probability of its copies, averaged over the different pools and the different copies.
Since both our algorithm and the greedy algorithm treat agents with equal feature vector symmetrically, averaging over the copies of an ESS participant is a valid way to estimate the end-to-end probability of any single copy, which greatly reduces sample variance.

In the body of the paper, we mention the behavior of the greedy algorithm without any quotas.
In this case, the panel members seem to be sampled with near-equal probability from the pool, which leads to end-to-end probabilities that are roughly proportional to $q_i$:

\includegraphics[width=\textwidth]{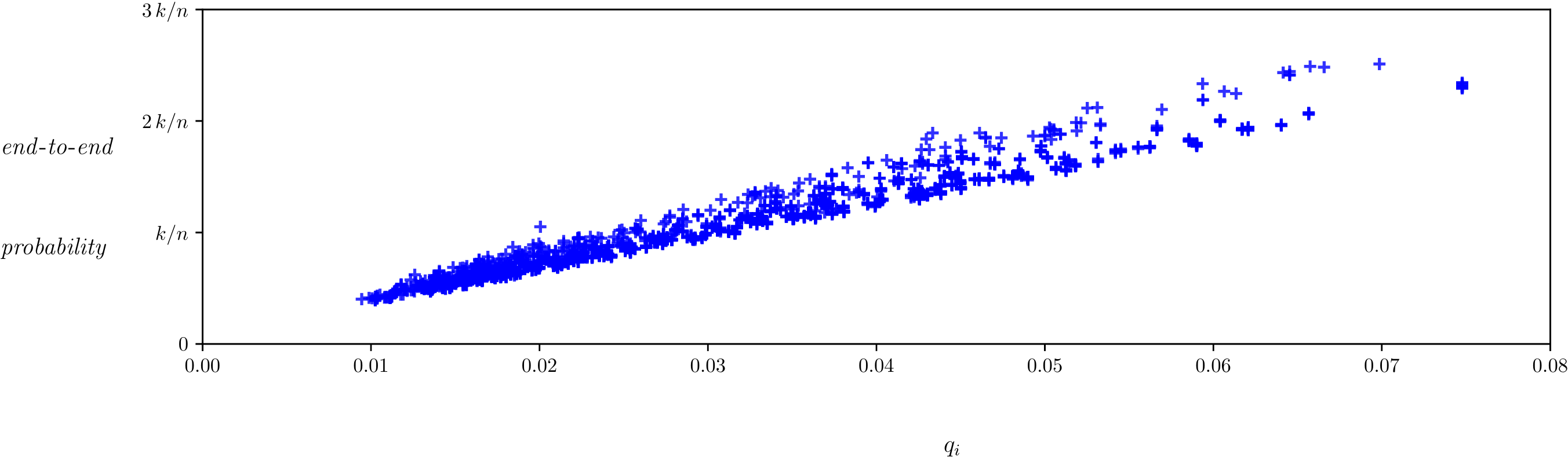}

\subsection{Additional Results for \Cref{sec:empirical}} \label{sec:app_experiments}
\paragraph{End-to-End Fairness Results for Varied $\boldsymbol{r}$ Values}
This plot shows the end-to-end probabilities for all agents in the synthetically-generated population over varied values of $r$.
To recall, we copied the agents in the background sample (in proportion to their weight) to obtain a synthetic population of size 60 million (the order of magnitude of eligible participants for the Climate Assembly).

We display these end-to-end probabilities for $r$ values 11\,000, 12\,000, 13\,000, and 60\,000, where 60\,000 is the $r$ value used to form the real-life Climate Assembly UK pool.
Every point in the scatter plot corresponds to an original member of the background sample, and the point's y-value is the mean selection probabilities averaged over 100\,000 sampled pools and over all copies of this background agent.\footnote{Averaging over the copies of an agent makes use of the fact that the selection process treats copies of the same agent symmetrically, which makes the empirical means converge faster.}

An important question is what we do when a bad pool occurs.
In the corresponding figure in the body of the text (examining only $r = 60\,000$), we did not credit any selection probability to any agent when bad pools occurred.
When we take this approach for multiple $r$ values, the result shows a sharp discontinuity between $r=11\,000$ (when everyone's end-to-end probability is essentially zero) and $r = 12\,000$ (when it is around $95\%$).
As it turns out, the property that makes nearly all pools bad when $r=11\,000$ is \cref{eq:good3}.
Note that this property is the least consequential of the three defining properties of a good pool: if we proceed with Part~II of the algorithm on a pool that satisfies only \cref{eq:good1,eq:good2}, we still satisfy the quotas but just can't bound the end-to-end probabilities.
Since the end-to-end probabilities are what we are measuring here anyway, we will in the following graph count bad pools as good pools if they only violate \cref{eq:good3}.

As shown in the figure below, we see a smooth transition towards the end-to-end guarantee, where higher values of $r$ give better guarantees.
The agents with the lowest selection probabilities are suffering most from low values of $r$, with their end-to-end probability trailing that of the majority of other agents.
From $r = 15\,000$ upwards, however, all agents in the population receive an end-to-end probability that is very close to $k/n$.
This threshold roughly coincides with the point at which $\alpha$ becomes larger than one.
\begin{center}
    \includegraphics[width=\textwidth]{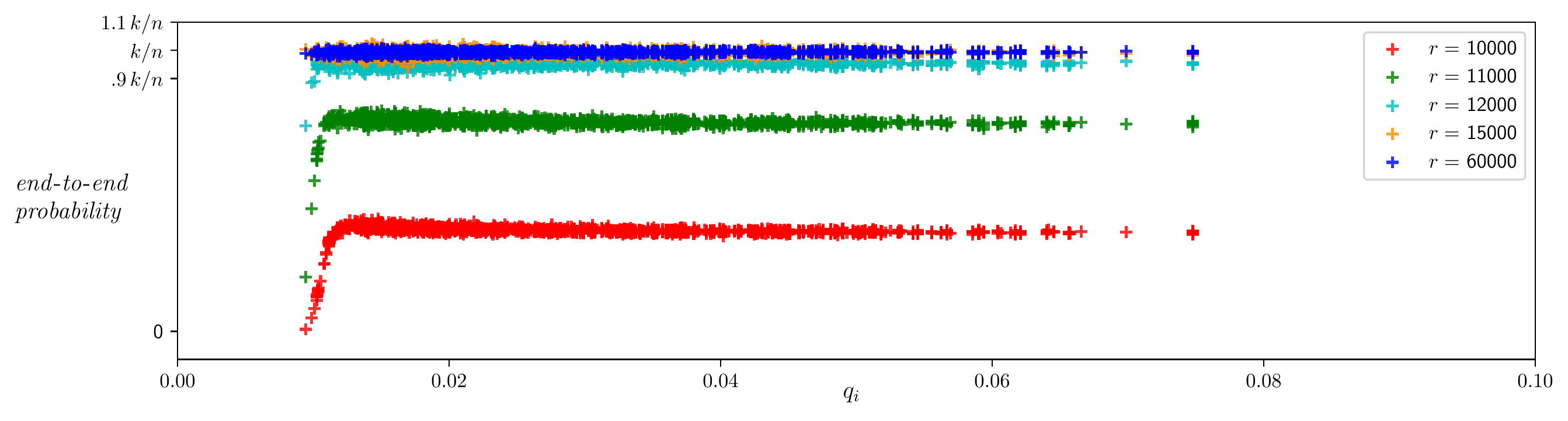}
\end{center}

\subsection{Validation and Results including Climate Concern Feature}
\label{sec:app_climate}
This section includes all the analysis in this paper and appendices, re-done with the climate concern level feature included. Figures in this section are provided in the same order as they were presented in the body of the paper,  \Cref{sec:app_betavalidation}, and \Cref{sec:app_experiments}.

\begin{center}
        \textit{\textbf{(Figures from Paper Body)}}
\end{center}

We omit the figure showing end-to-end probabilities at $r=60,000$, because when the \textit{Climate Concern Level} feature is included, good pools are so rare at this value of $r$ that all end-to-end probabilities are 0.
Similarly, for the greedy algorithm, the floor and ceiling quotas are often not satisfiable. In 754 out of 1\,000 random pools, this is because fewer pool members are ``not at all concerned'' about climate change than the lower quota for this feature, which is 6. In 86 out of the remaining pools, the greedy algorithm fails to identify a valid panel within the first 100 restarts. Only in the remaining 160 pools did the greedy algorithm find a valid panel in fewer than 100 iterations.

\begin{center}
    \hspace*{0.45cm}\includegraphics[width=0.8\textwidth]{realized_representation_legend.pdf}\\[-.5em]
    \includegraphics[width=\textwidth]{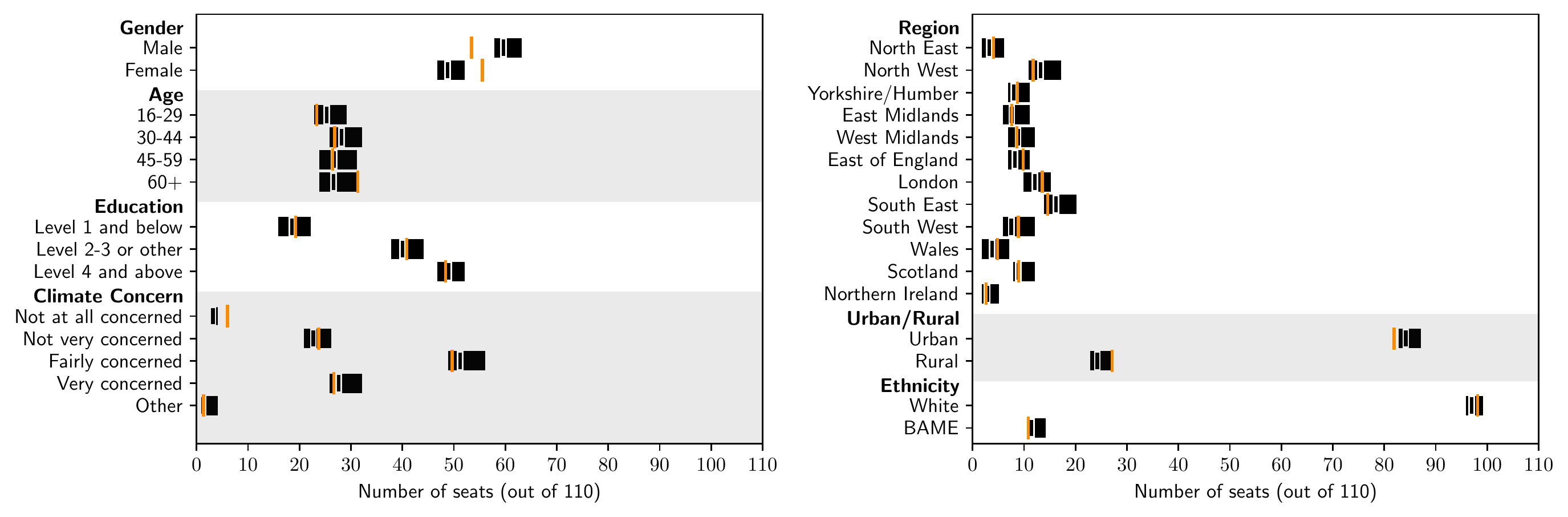}
\end{center}

\pagebreak
\begin{center}
    \textit{\textbf{(Figures from \Cref{sec:app_betavalidation})}}
\end{center}

\paragraph{Pool and Background Data Composition}
\begin{center}
    \includegraphics[width=\textwidth]{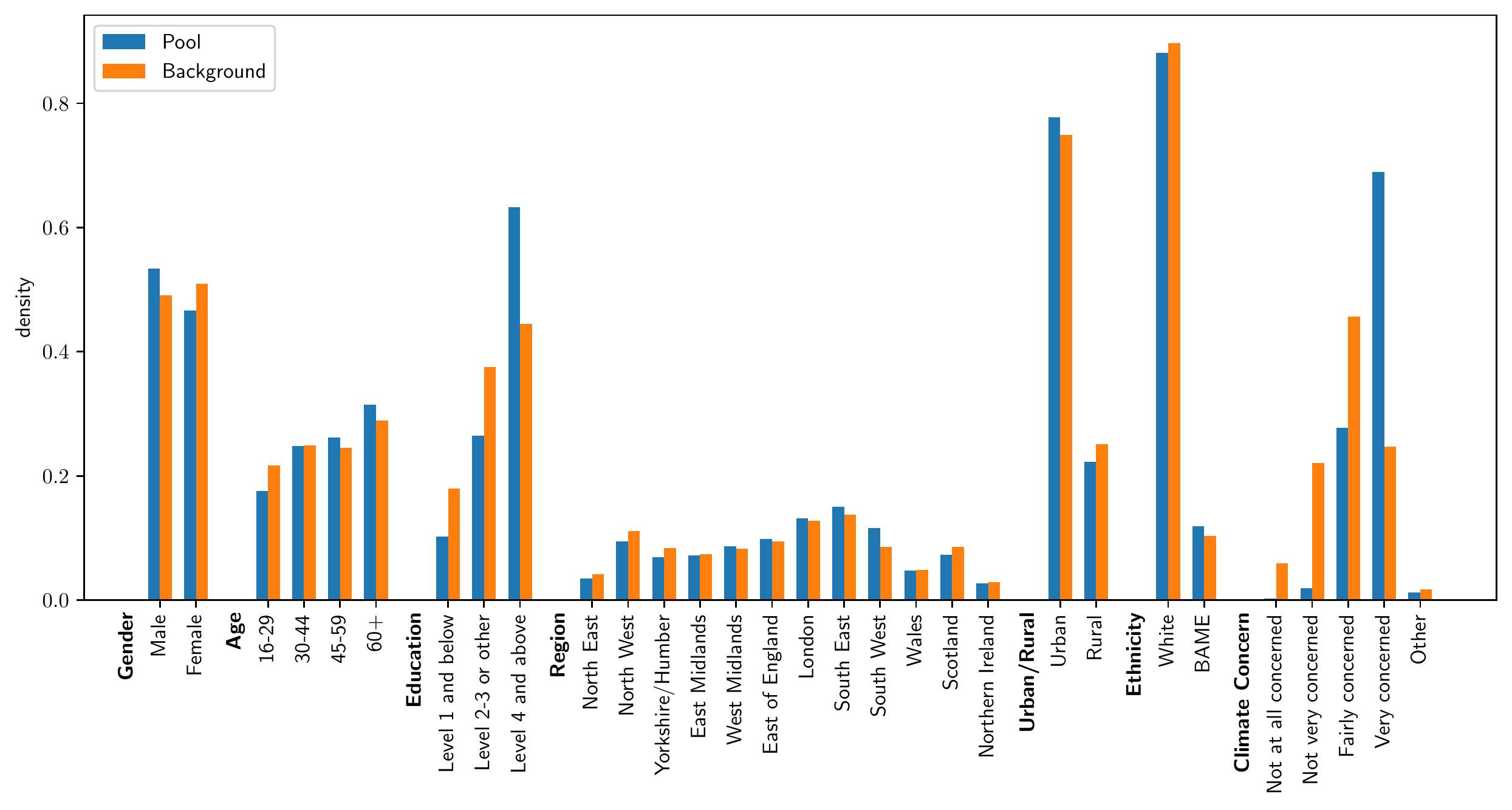}
\end{center}

\paragraph{Estimates of $\boldsymbol{\beta}$}
\begin{center}
    \includegraphics[width=\textwidth]{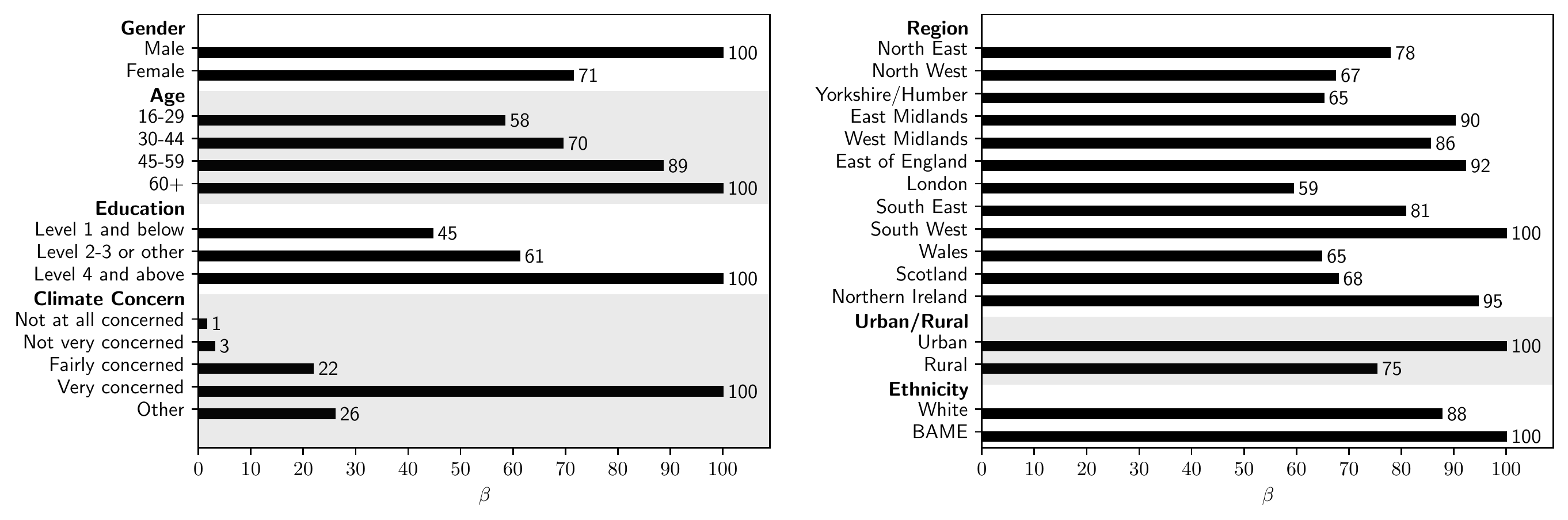}
\end{center}

\paragraph{Estimates of $\boldsymbol{q_i}$} $\beta_0 = 24.3\%$. Frequencies of $q_i$ values above 15\% are not shown due to privacy concerns. 6.8\%, 1\% of agents in pool, background datasets respectively are not presented for this reason.
\begin{center}
    \includegraphics[width=0.6\textwidth]{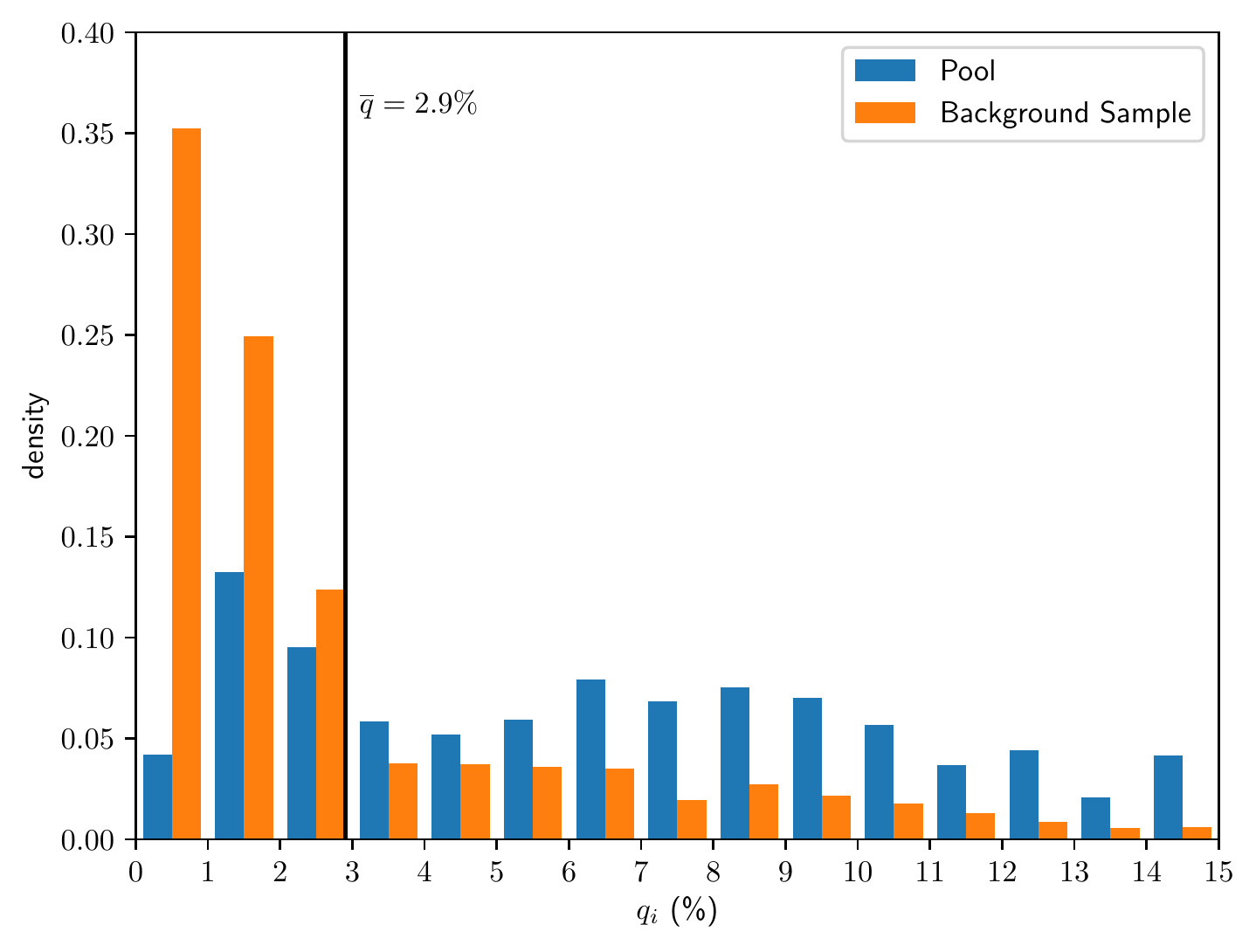}
\end{center}

\paragraph{Test for calibration of $q_i$ estimates} Frequencies of $q_i$ values above 20\% are not shown due to privacy concerns. Less than 0.4\% of agents in either dataset are not presented for this reason.
\begin{center}
    \includegraphics[width=0.6\textwidth]{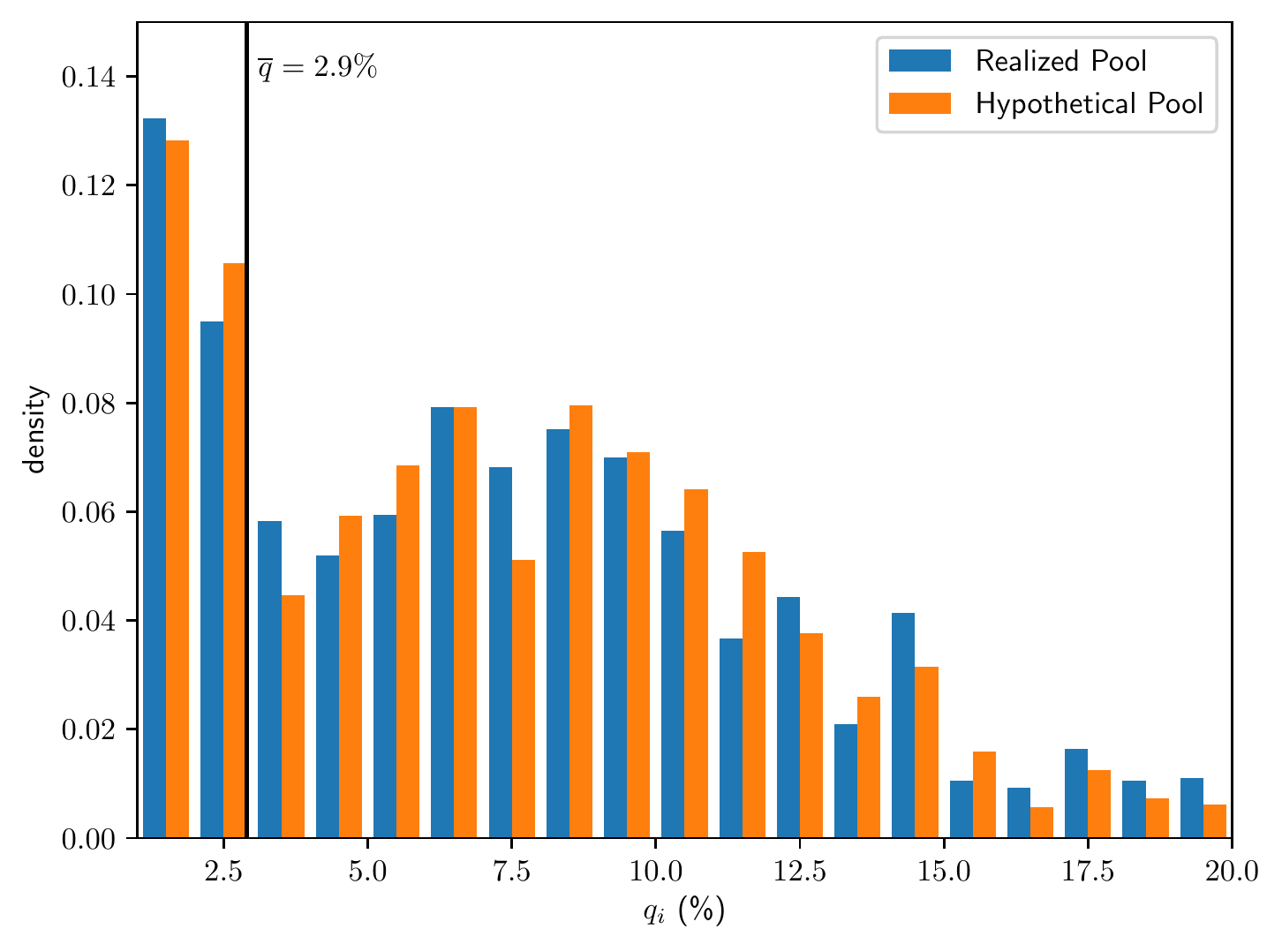}
\end{center}

\paragraph{Comparison of Realized Pool Composition and Hypothetical Pool Composition}

\begin{center}
    \includegraphics[width=\textwidth]{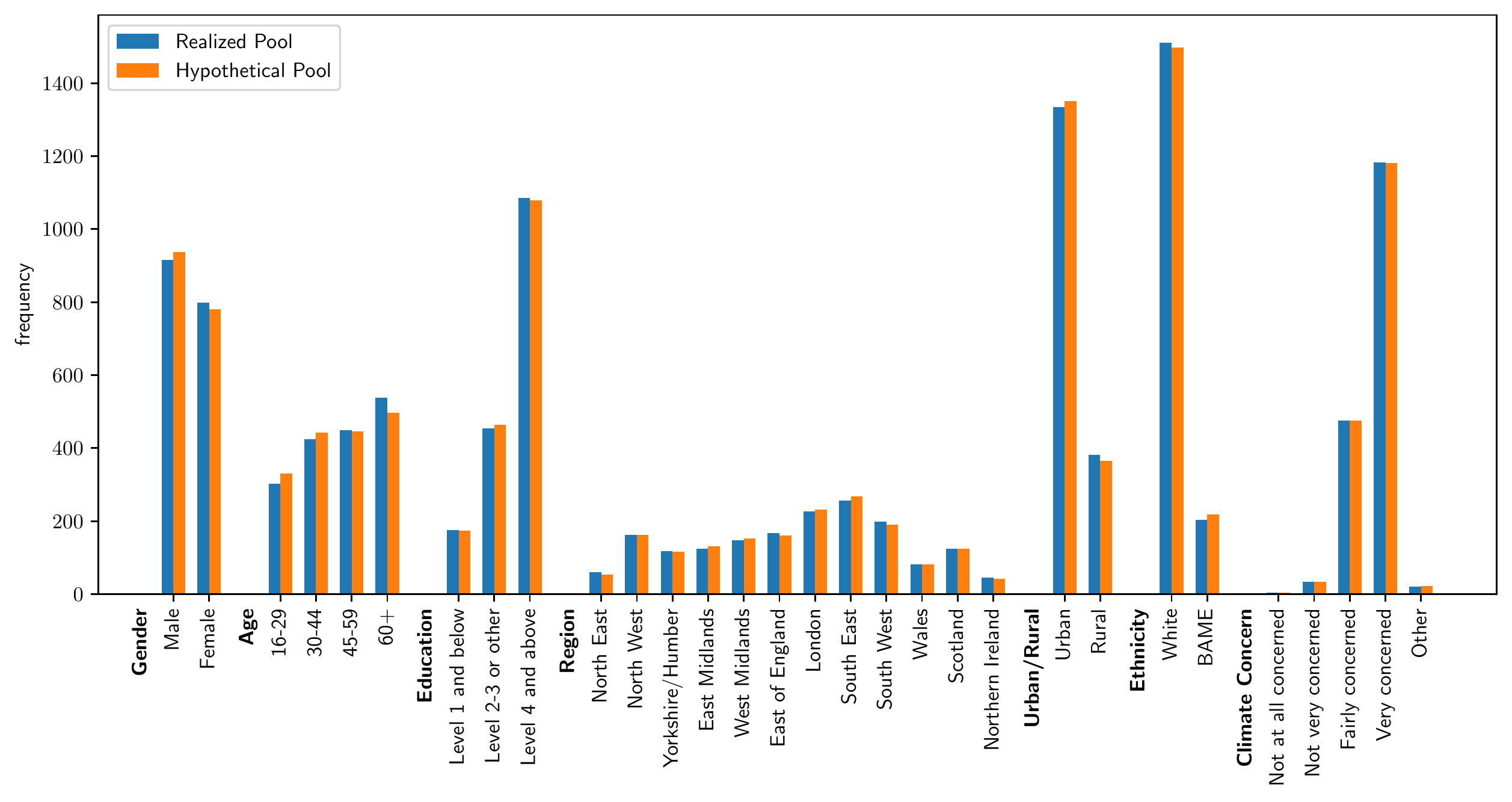}
\end{center}

\paragraph{Testing model capture of 2-correlations}
\begin{center}
    \includegraphics[width=.75\textwidth]{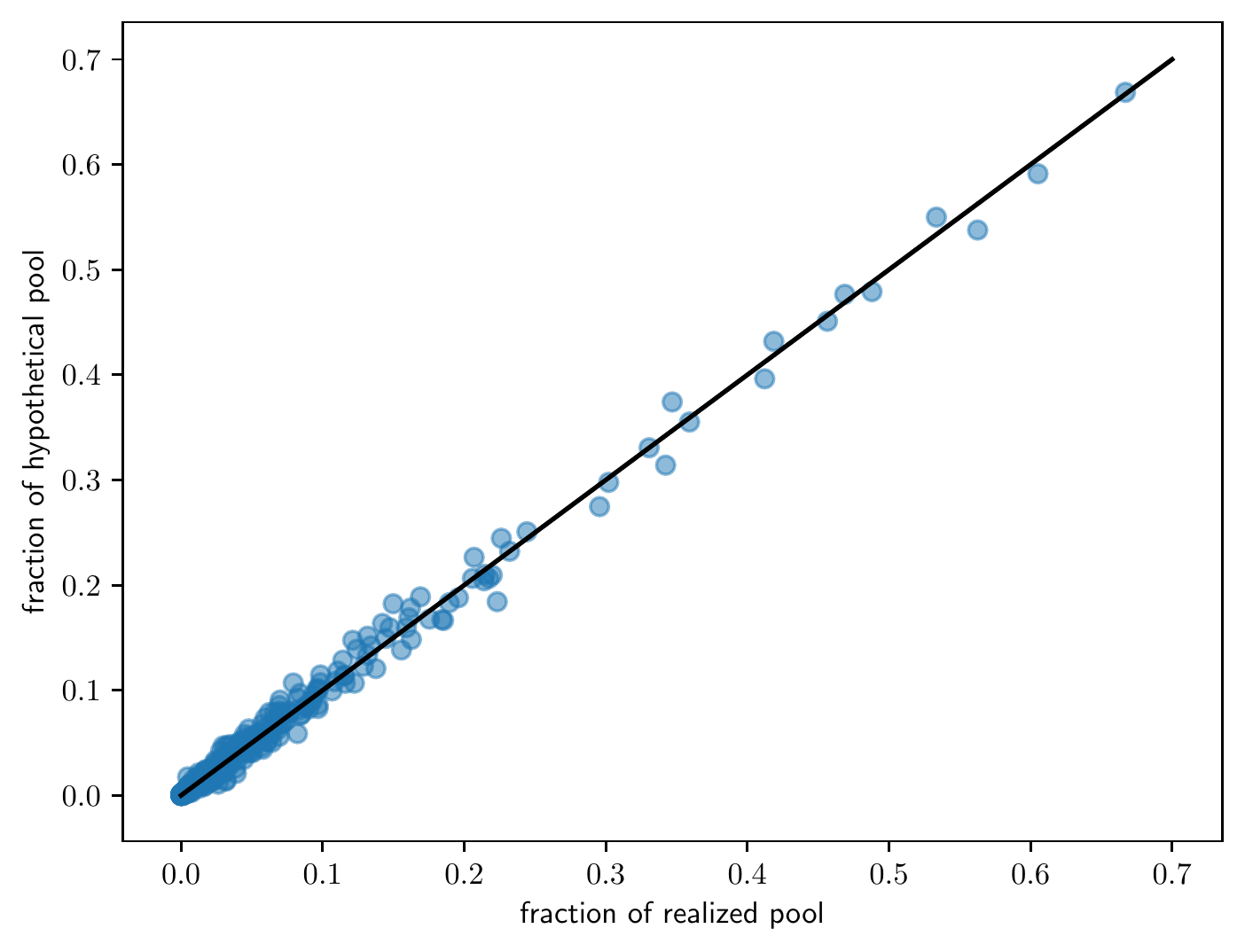}
\end{center}

\begin{center}
    \textit{\textbf{(Figures from \Cref{sec:app_experiments})}}
\end{center}

\paragraph{End-to-End Fairness Results for Varied $\boldsymbol{r}$ Values} This figure demonstrates that, for large enough $r$, we can get $k/n$ end-to-end probability for all agents in the synthetic population when we include the Climate Concern Level feature. We only include analysis for only one $r$ value because the $r$ values must be extremely large to give any end-to-end guarantees when the Climate Concern Feature is included, and running the analysis with such large $r$ costs substantial computational time.
\begin{center}
    \includegraphics[width=\textwidth]{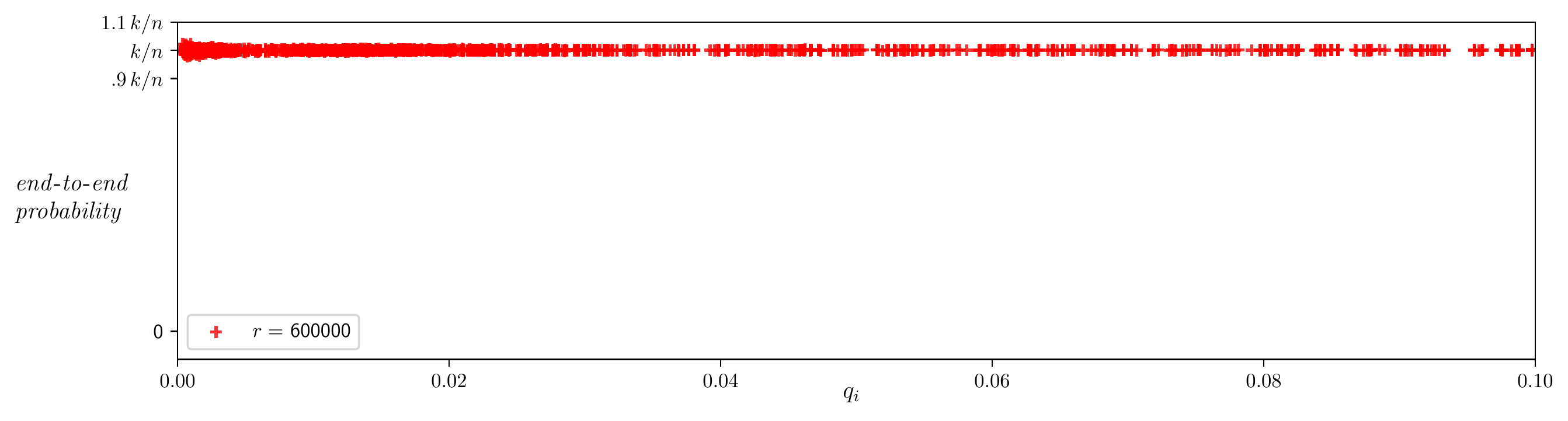}
\end{center}

\end{document}